\documentclass[aps,prd,preprint,titlepage,eqsecnum,nofootinbib,showpacs,showkeys,amsmath,amssymb]{revtex4-2}

\usepackage[utf8]{inputenc}
\usepackage[T1]{fontenc}
\usepackage[english]{babel}
\usepackage{amsfonts,bm,graphicx,xcolor}
\usepackage{theorem}
\usepackage{etoolbox}
\usepackage{hyperref}

\theorembodyfont{\normalfont}
\newtheorem{definition}{Definition}[section]

\theorembodyfont{\itshape}
\newtheorem{lemma}[definition]{Lemma}

\newenvironment{proof}[1][Proof]
  {\par\addvspace{\medskipamount}\noindent\textit{#1.}\ignorespaces}
  {\unskip\nobreak\hfill\ensuremath{\square}\par\addvspace{\medskipamount}}

\hypersetup{
  pdftitle={Coherent states of neutrinos in rotating matter},
  pdfauthor={Alexander Breev and Maxim Dvornikov},
  pdfkeywords={coherent states, neutrinos, Dirac equation, light-cone variables, rotating matter}
}

\begin{document}

\title{Coherent states of neutrinos in rotating matter}

\author{Alexander Breev}
\email{alexphys90@gmail.com}

\author{Maxim Dvornikov}
\email{maxim.dvornikov@gmail.com}

\affiliation{Pushkov Institute of Terrestrial Magnetism, Ionosphere and Radiowave Propagation (IZMIRAN), 108840 Moscow, Troitsk, Russia}

\begin{abstract}
We study coherent states of neutrinos electroweakly interacting with rotating matter. We start with the Dirac equation for a massive neutrino propagating in background matter with a coordinate dependent effective potential. Then, we consider the limit of massless particles and the specific matter velocity profile. It allows us to reduce the neutrino Hamiltonian to the harmonic oscillator form. We find the exact solutions of the Dirac equation and construct the creation and annihilation operators. Using these results, we build the generalized coherent states of neutrinos. We provide the interpretation of these coherent states in terms of the classical physics. In particular, we demonstrate that the coherent states constructed approach plane waves in the limit of slow matter rotation. We also define the effective neutrino mass owing to the interaction of a massless neutrino with rotating matter. Then, we consider a special case of the generalized coherent states which are called the neutrino coherent states. The situation of the degenerate parameters of the system is also studied. Finally, we apply our results for astrophysical neutrinos propagating in dense matter of a rapidly rotating neutron star (NS). We calculate the effective neutrino masses for different neutrino energies for realistic parameters of NS.
\end{abstract}

\pacs{13.15.+g, 03.65.Pm, 03.65.Ge}
\keywords{coherent states, squeezed coherent states, neutrino interaction with matter,
Dirac equation, light-cone variables, rotating matter, neutron stars}

\maketitle

\section{Introduction}\label{sec:lNTR}

Neutrinos are elusive particles which interact rather weakly with
external fields. It happens, e.g., since the cross section of the
medium energy neutrinos interaction with other particles is quite
small (see, e.g., Ref.~\cite{ForZel12}). This fact results in significant
difficulties in the neutrino detection in terrestrial conditions since
existing and future neutrino detectors and telescopes should have quite large
working volumes. Despite these difficulties, the prospects of the
low energy relic neutrinos detection were put forward recently (see
Ref.~\cite{She21} for a review)

We also mention that neutrinos are not excluded to have nontrivial
electromagnetic properties~\cite{Giu25}. However, the magnitude
of neutrino magnetic moments is likely to be quite small. It makes
the neutrino electromagnetic interaction to be subdominant compared
to the matter one. The recent proposals for the studies of neutrino
electromagnetic properties are present in Ref.~\cite{GiuTer23}.

Nevertheless, in some situations, neutrino interaction with matter
plays a crucial role in the evolution of these particles. It happens
mainly for astrophysical neutrinos. Neutrinos emitted in a core collapse
of a supernova are supposed to drive the explosion process~\cite{BetWil85}.
At the some stages of the young neutron star (NS) life, neutrinos
contribute to the star cooling~\cite{Lat91}.

Low energy neutrinos, created in NS, were suggested in Refs.~\cite{Loe90,KieWei97}
to be captured owing to the interaction with dense matter of NS. The
trapping of neutrinos, predicted in Ref.~\cite{Loe90}, is owing to
the inhomogeneous matter distribution inside NS. The kinetic approach
to neutrinos in inhomogeneous matter was developed in Ref.~\cite{Sil99}.

The coordinate dependent neutrino interaction potential can be also
implemented in a rotating matter. The neutrino quantum states electroweakly
interacting with rotating matter were found in Refs.~\cite{GrigorievSavochkinStudenikin2007,Dvornikov2014,BreevDvornikov2026}.
The studies in Refs.~\cite{GrigorievSavochkinStudenikin2007,Dvornikov2014,BreevDvornikov2026} were based
on the exact solution of the Dirac equation in the presence of the
axial vector external field. One of the difficulties in solving of
this kind of Dirac equation is the consideration of the nonzero neutrino
mass. That is why, in the problem in question, one considers either
massless neutrinos~\cite{GrigorievSavochkinStudenikin2007} or takes into account the neutrino
mass perturbatively~\cite{BalPopStu11}.

Another problem one encounters while studying neutrinos in rotating
matter is a smooth limit to the free particles case. The quantum mechanical
description of neutrinos in rotating matter was shown in Ref.~\cite{GrigorievSavochkinStudenikin2007}
to be equivalent to a charged particle motion in a magnetic field
$B_{\mathrm{eff}}\propto\omega$, where $\omega$ is the angular velocity
of matter. We have already mentioned that the neutrino interaction
with matter is quite weak. Moreover, angular velocities are not so
high even for rapidly rotating pulsars. However, if we consider the
limit of a slowly rotating matter, $\omega\to0$, the neutrino wavefunction
does not approach to a plane wave.

This difficulty can be overcome by considering coherent states as special localized
solutions of the single particle Dirac equation instead of stationary solutions. This approach allows us to
study the limit $B\to0$ and obtain a correct transition to plane-wave
solutions.

Coherent states were first
proposed in Ref.~\cite{Gla63} to describe the evolution of a many
bosons system. Coherent states of a nonrelativistic electron in a
magnetic field were built in Ref.~\cite{FelKah70}. The limit of
a zero magnetic field was considered in Ref.~\cite{Var84}.

The construction of coherent states for a relativistic electron in
a magnetic field is not so trivial. The main difficulty is that the
different components of a bispinor, which is the solution of the Dirac
equation, depend on $n$ and $n-1$, where $n$ is the principle quantum
number. Therefore, the reduction of the Hamiltonian of such an electron
to the harmonic oscillator form is not straightforward. Nevertheless,
an attempt to construct the coherent states of a relativistic electron
in a magnetic field was made in Ref.~\cite{MalkinManko1969}.

The theory of coherent states and their various generalizations has been extensively developed in quantum mechanics and quantum optics. The group theory approach to generalized coherent states is presented in Ref.~\cite{Perelomov1986}. The review of the coherent-state theory and its applications is given in Ref.~\cite{ZhangFengGilmore1990}. The modern treatment of the coherent-state theory can be found in the Ref.~\cite{Gazeau2009}. Squeezed and other nonclassical states of light are discussed in detail in Refs.~\cite{DodonovManko2003,BagchiGhoshKhare2020}.

In our recent works, we have developed various methods for constructing and generalizing coherent states. In Ref.~\cite{ShapovalovBreev2023}, coherent states of a harmonic oscillator were considered from the standpoint of the non-commutative integration method and the geometry of coadjoint orbits of a Lie group. In Ref.~\cite{BreevGitmanDerolle2024}, generalized coherent states were constructed for a massive particle under the influence of a constant force. This result extended the approach to a system with unbounded motion and a continuous energy spectrum. In Ref.~\cite{BreevGitman2025}, we proposed a method for constructing coherent states of finite-dimensional quantum systems with a given angular momentum. Finally, in Ref.~\cite{BreevGitman2026}, we established a relation between generalized coherent states and solutions of the Schr\"odinger equation on Lie groups obtained by the non-commutative integration method.

Coherent states that are exact solutions of relativistic wave equations have been studied for the Klein--Gordon and the Dirac equations in external electromagnetic fields. In Ref.~\cite{BagrovBuchbinderGitman1976}, such states were constructed using the null-plane formulation of relativistic quantum mechanics. Solutions of the Klein--Gordon and Dirac equations with squeezed coherent states of transverse motion in a plane electromagnetic wave were obtained in Ref.~\cite{GitmanBagrov2005}. Coherent and semiclassical states of relativistic particles in a magnetic-solenoid field, including solutions of the Dirac equation in $(3+1)$ dimensions in light-cone variables, were studied in Ref.~\cite{BagrovGavrilovGitmanMeira2011}.

In the present work, we generalize the method for constructing coherent states of the Dirac equation in light-cone variables to the Dirac equation for a neutrino electroweakly interacting with moving matter. This paper is organized in the following way. In Sec.~\ref{sec:lf-reduction}, we formulate the light-cone formalism for a neutrino interacting with a coordinate dependent external field. Taking that neutrinos are massive particles, we consider a particular case of background matter rotating with a constant angular velocity in Sec.~\ref{sec:lf-reduction}. Then, in Sec.~\ref{sec:oscillator-representation}, we solve the Dirac equation in the light-cone representation for massless active neutrinos. The generalized coherent states are built in Sec.~\ref{sec:coherent-states} based on the solutions found in Sec.~\ref{sec:oscillator-representation}. In Sec.~\ref{sec:classical-limits}, we interpret the obtained generalized coherent states in terms of the classical physics and consider the limit of a slowly rotating matter resulting in the free motion of a neutrino. In Sec.~\ref{sec:coherent-classification}, we select particular states, which, in fact are called the neutrino coherent states, among the general ones derived Sec.~\ref{sec:coherent-states}. The neutrino coherent states in the peculiar case of the parameters of the system in studied in Sec.~\ref{sec:singular-lambda}. In Sec.~\ref{sec:lambda-neutron-star}, we apply our results to astrophysical neutrinos emitted by a rotating NS. In particular, we compute the effective mass of a neutrino owing to the electroweak interaction with rotating matter. Finnaly, we conclude in Sec.~\ref{sec:CONCL}. Some mathematical details are provided in Appendix~\ref{app-A}. 

\section{The Dirac equation for neutrinos in light-cone variables}\label{sec:lf-reduction}

In this section, we study the Dirac equation (see Eq.~\eqref{eq2.2} below) in the light-cone variables for a neutrino electroweakly interacting with moving matter.

We consider a single massive neutrino without mixing. The neutrino is supposed to be a Dirac particle. The interaction with matter is taken in the forward scattering approximation with the following effective matter potential:
\begin{equation}
	f^\mu(x)
	=
	-\frac{G_N}{\sqrt2}
	\left[1,v^1(x^1,x^2),v^2(x^1,x^2),v^3(x^1,x^2)\right],\quad
	G_N=G_\mathrm{F} n,
	\label{eq2.1}
\end{equation}
where $(x^1,x^2,x^3)$ are local Cartesian coordinates, $(v^1,v^2,v^3)$ are the matter velocity components, $G_\mathrm{F} =1.17\times10^{-5}\,\text{GeV}^{-2}$ is the Fermi constant, and $n$ is the effective matter density. The explicit form for $n$ for different neutrino types can be found in Ref.~\cite{DvoStu02}.

Let us write down the Dirac equation for a neutrino interacting with moving matter, accounting for Eq.~\eqref{eq2.1}, in covariant form:
\begin{equation}
\left[
\gamma^\mu \hat p_\mu
+\frac12\gamma_\mu(1+\gamma^5)f^\mu
-m_\nu
\right]\Psi(x)=0,
\qquad
\hat p_\mu=\mathrm{i}\frac{\partial}{\partial x^\mu}.
\label{eq2.2}
\end{equation}
Here $m_\nu$ is the neutrino mass and $\gamma^\mu$ are the Dirac gamma matrices corresponding to the metric signature $(+---)$,
\begin{equation}
	\{\gamma^\mu,\gamma^\nu\}=2g^{\mu\nu}I_4,\quad
	\left(g^{\mu\nu}\right)=\operatorname{diag}(1,-1,-1,-1),\quad
	\mu,\nu = 0,\dots,3.
	\label{eq2.3}
\end{equation}
The operator $(1+\gamma^5)/2$ selects active neutrino states, i.e., the left-chiral subsystem.

Let us introduce light-cone variables
\begin{equation}
\begin{aligned}
&u^0=\frac{x^0-x^3}{\sqrt2},\qquad
u^1=x^1,\qquad
u^2=x^2,\qquad
u^3=\frac{x^0+x^3}{\sqrt2},\\
&x^0=\frac{u^0+u^3}{\sqrt2},\qquad
x^1=u^1,\qquad
x^2=u^2,\qquad
x^3=\frac{u^3-u^0}{\sqrt2}.
\end{aligned}
\label{eq2.4}
\end{equation}
We call the hypersurfaces $u^0=\mathrm{const}$ null hyperplanes. Shortly, $u^0$ is used as the evolution parameter. This choice of the evolution parameter was proposed by Dirac in Ref.~\cite{Dirac1949}. The metric in the variables $u^\mu$ takes the form
\begin{equation}
\begin{aligned}
&\left(\widetilde g^{\mu\nu}\right)=
\left(\widetilde g_{\mu\nu}\right) =
\begin{pmatrix}
0&0&0&1\\
0&-1&0&0\\
0&0&-1&0\\
1&0&0&0
\end{pmatrix},\quad
ds^2=2\,du^0du^3-(du^1)^2-(du^2)^2.
\end{aligned}
\label{eq2.5}
\end{equation}
The momentum operators in the light-cone variables have the form,
\begin{equation}
\widehat{\widetilde p}_\mu=\mathrm{i}\frac{\partial}{\partial u^\mu},
\qquad
\widehat{\widetilde p}_3=\frac{\hat p_0+\hat p_3}{\sqrt2},\quad
\widehat{\widetilde p}_\pm=\widehat{\widetilde p}_1\pm\mathrm{i}\widehat{\widetilde p}_2.
\label{eq2.6}
\end{equation}

The Dirac gamma matrices in light-cone variables are defined by
\begin{align}
&\widetilde\gamma^0
=\frac{\gamma^0-\gamma^3}{\sqrt2},
\qquad
\widetilde\gamma^1=\gamma^1,
\qquad
\widetilde\gamma^2=\gamma^2,
\qquad
\widetilde\gamma^3
=\frac{\gamma^0+\gamma^3}{\sqrt2}.
\label{eq2.7}\\
&\{\widetilde\gamma^\mu,\widetilde\gamma^\nu\}
=2\widetilde g^{\mu\nu}I_4,\quad 
(\widetilde\gamma^0)^2=(\widetilde\gamma^3)^2=0,
\qquad
\{\widetilde\gamma^0,\widetilde\gamma^3\}=2I_4,\quad
\widetilde\gamma^5=\gamma^5.\label{eq2.8}
\end{align}
The covariant components of the effective potential~\eqref{eq2.1} in light-cone variables take the form,
\begin{equation}
\begin{aligned}
&\widetilde f_\mu(u)
=f_\lambda(x(u))\frac{\partial x^\lambda}{\partial u^\mu}
=G_N\left(
-\frac{1+v^3}{2},
\frac{v^1}{\sqrt2},
\frac{v^2}{\sqrt2},
-\frac{1-v^3}{2}
\right).
\end{aligned}
\label{eq2.9}
\end{equation}
Thus, the Dirac Eq.~\eqref{eq2.2} in light-cone variables can be written as
\begin{equation}
\begin{aligned}
&\left[
\widetilde\gamma^\mu\widehat{\widetilde p}_\mu+U(u)-m_\nu
\right]\Psi(u)=0,\quad
U(u)=\frac12\widetilde\gamma_\mu(1+\gamma^5)\widetilde f^{\,\mu}(u).
\end{aligned}
\label{eq2.10}
\end{equation}

Let us introduce mutually orthogonal projectors,
\begin{equation}
\widehat P_{(-)}
=
\frac12\widetilde\gamma^3\widetilde\gamma^0,
\qquad
\widehat P_{(+)}
=
\frac12\widetilde\gamma^0\widetilde\gamma^3.
\label{eq2.11}
\end{equation}
They satisfy the relations,
\begin{equation}
\widehat P_{(+)}+\widehat P_{(-)}=I_4,
\qquad
\widehat P_{(\pm)}^2=\widehat P_{(\pm)},
\qquad
\widehat P_{(+)}\widehat P_{(-)}=0,
\qquad
\widehat P_{(\pm)}^\dagger=\widehat P_{(\pm)}.
\label{eq2.12}
\end{equation}
Let us decompose the bispinor,
\begin{equation}
\Psi=\Psi_{(+)}+\Psi_{(-)},
\qquad
\Psi_{(\pm)}=\widehat P_{(\pm)}\Psi.
\label{eq2.13}
\end{equation}
The projection decomposition of the Dirac equation for an electron in light-cone variables, including the separation of the independent projected component and the reconstruction of the remaining bispinor components, is developed in Ref.~\cite{BagrovGitman2014}.

\begin{lemma}\label{lem:projected-reduction}
The Dirac Eq.~\eqref{eq2.10} is equivalent to a system of two projected equations:
\begin{align}
&\left(\widehat{\widetilde p}_0+\frac12\widetilde\gamma^3\widehat P_{(+)}U\right)\Psi_{(-)}
=\frac12\left[(\widetilde\gamma^j\widehat{\widetilde p}_j+m_\nu)\widetilde\gamma^3
-\widetilde\gamma^3\widehat P_{(+)}U\right]\Psi_{(+)}, \label{eq2.14} \\
&\left[2\widehat{\widetilde p}_3+\widetilde\gamma^0\widehat P_{(-)}U\right]\Psi_{(+)}
=\left[(\widetilde\gamma^j\widehat{\widetilde p}_j+m_\nu)\widetilde\gamma^0
-\widetilde\gamma^0\widehat P_{(-)}U\right]\Psi_{(-)},\quad j=1,2.
\label{eq2.15}
\end{align}
If the inverse operator $\left[2\widehat{\widetilde p}_3
+\widetilde\gamma^0\widehat P_{(-)}U\right]^{-1}$ exists, the original Dirac Eq.~\eqref{eq2.10} is equivalent to the closed Eq.~\eqref{eq2.16} for the component $\Psi_{(-)}$,
\begin{equation}
\begin{aligned}
&\Bigg\{
2\widehat{\widetilde p}_0
+\widetilde\gamma^3\widehat P_{(+)}U\\
&-\left[
(\widetilde\gamma^j\widehat{\widetilde p}_j+m_\nu)\widetilde\gamma^3
-\widetilde\gamma^3\widehat P_{(+)}U
\right]
\left[
2\widehat{\widetilde p}_3+\widetilde\gamma^0\widehat P_{(-)}U
\right]^{-1}\\
&\times\left[
(\widetilde\gamma^j\widehat{\widetilde p}_j+m_\nu)\widetilde\gamma^0
-\widetilde\gamma^0\widehat P_{(-)}U
\right]
\Bigg\}\Psi_{(-)}=0,
\qquad j=1,2,
\end{aligned}
\label{eq2.16}
\end{equation}
whereas the component $\Psi_{(+)}$ is given by
\begin{equation}
\Psi_{(+)}=
\left[2\widehat{\widetilde p}_3+\widetilde\gamma^0\widehat P_{(-)}U\right]^{-1}
\left[(\widetilde\gamma^j\widehat{\widetilde p}_j+m_\nu)\widetilde\gamma^0
-\widetilde\gamma^0\widehat P_{(-)}U\right]\Psi_{(-)}.
\label{eq2.17}
\end{equation}
\end{lemma}

The proof of this lemma is given in Appendix~\ref{app-A}.

Since the potential in Eq.~\eqref{eq2.1} is independent of $u^3$, we have
\begin{equation}
\left[
\widehat{\widetilde p}_3,
\widetilde\gamma^\mu\widehat{\widetilde p}_\mu+U(u)-m_\nu
\right]=0.
\label{eq2.18}
\end{equation}
Hence, the solutions of the Dirac Eq.~\eqref{eq2.10} can be looked for among the eigenfunctions of the symmetry operator $\widehat{\widetilde p}_3$:
\begin{equation}
\widehat{\widetilde p}_3\Psi=\frac{\lambda}{2}\Psi,\qquad \lambda\in\mathbb R.
\label{eq2.19}
\end{equation}
If a solution is simultaneously an eigenstate of the operators $\hat p_0$ and $\hat p_3$ with eigenvalues $p_0$ and $p_3$, Eq.~\eqref{eq2.6} implies the physically useful relation
\begin{equation}
\lambda=\sqrt2\,(p_0+p_3).
\label{eq2.20}
\end{equation}
Here $p_0=E$ and $p_3$ are eigenvalues of the covariant components of the four-momentum. For the signature $(+---)$, we have $p_3=-p^3$, i.e.
$\lambda=\sqrt2\,(E-p^3)$. 

Eq.~\eqref{eq2.19} implies that
\begin{equation}
\Psi(u)
=
\psi(u^0,u^1,u^2)
\exp\left(-\frac{\mathrm{i}\lambda}{2}u^3\right).
\label{eq2.21}
\end{equation}
For solutions of the form given in Eq.~\eqref{eq2.21} with a fixed value of the quantum number $\lambda$, integration over $u^3$ gives
\begin{equation}
	\int_{\mathbb R}
	\exp\!\left[\frac{\mathrm{i}(\lambda_1-\lambda_2)}{2}u^3\right]\mathrm{d} u^3
	=4\pi\delta(\lambda_1-\lambda_2).
	\label{eq2.22}
\end{equation}

The action of the inverse operator on the functions in Eq.~\eqref{eq2.21} is readily calculated,
\begin{equation}
\begin{aligned}
\left[2\widehat{\widetilde p}_3+\widetilde\gamma^0\widehat P_{(-)}U\right]^{-1}\Psi(u) = 
\left[\lambda+\widetilde\gamma^0\widehat P_{(-)}U\right]^{-1}\Psi(u).
\end{aligned}
\label{eq2.23}
\end{equation}
This operator exists provided that
\begin{equation}
\lambda\neq0,
\qquad
\lambda-G_N(1-v^3)\neq0.
\label{eq2.24}
\end{equation}
For $v^3=v^3(u^1,u^2)$, the second condition in Eq.~\eqref{eq2.24} should hold at every point of the domain in $(u^1,u^2)$ where the solution is constructed.
Below, the reduction Eqs.~\eqref{eq2.16} and~\eqref{eq2.17} are used only when the conditions in Eq.~\eqref{eq2.24} are valid.

To rewrite the matrix Eqs.~\eqref{eq2.14}--\eqref{eq2.17} and their solutions in component form, let us fix the chiral representation of the Dirac gamma matrices consistent with our sign convention for $\gamma^5$:
\begin{equation}
\begin{aligned}
&\gamma^0=
\begin{pmatrix}0&I_2\\ I_2&0\end{pmatrix},
\qquad
\gamma^k=
\begin{pmatrix}0&-\sigma_k\\ \sigma_k&0\end{pmatrix},
\qquad k=1,2,3,\\
&\gamma^5=
\begin{pmatrix}I_2&0\\0&-I_2\end{pmatrix},
\qquad
\widehat P_{(-)}=\operatorname{diag}(0,1,1,0),
\qquad
\widehat P_{(+)}=\operatorname{diag}(1,0,0,1).
\end{aligned}
\label{eq2.25}
\end{equation}
Here $\sigma_k$ are the Pauli matrices. In this representation, the projected bispinor components have the form
\begin{equation}
\Psi_{(-)}=
\begin{pmatrix}0\\ \psi_L\\ \psi_R\\0\end{pmatrix},
\qquad
\Psi_{(+)}=
\begin{pmatrix}\phi_L\\0\\0\\\phi_R\end{pmatrix}.
\label{eq2.26}
\end{equation}
The components $\psi_L$ and $\psi_R$ are taken to be independent on the hypersurface $u^0=\mathrm{const}$, whereas the components $\phi_L$ and $\phi_R$ are determined by Eq.~\eqref{eq2.15}.

Let us introduce a scalar product on the hypersurface \texorpdfstring{$u^0=\mathrm{const}$}{u0 = const}. For two solutions $\Psi_1$ and $\Psi_2$ of the Dirac Eq.~\eqref{eq2.10}, we define the mixed current density,
\begin{equation}
j^\mu[\Psi_1,\Psi_2]=\overline\Psi_1\gamma^\mu\Psi_2,
\qquad
\partial_\mu j^\mu[\Psi_1,\Psi_2]=0.
\label{eq2.27}
\end{equation}
On the hypersurface $u^0=\mathrm{const}$, the scalar product of solutions is defined by integrating the component $\widetilde j^{\,0}=\overline\Psi_1\widetilde\gamma^0\Psi_2$ of the mixed current in Eq.~\eqref{eq2.27},
\begin{equation}
(\Psi_1,\Psi_2)_{u^0}
=
\int_{\mathbb R^3}
\overline\Psi_1\widetilde\gamma^0\Psi_2\,
\mathrm{d} u^1\mathrm{d} u^2\mathrm{d} u^3.
\label{eq2.28}
\end{equation}
The identity $\gamma^0\widetilde\gamma^0=\sqrt2\,\widehat P_{(-)}$ implies that
\begin{equation}
(\Psi_1,\Psi_2)_{u^0}
=
\sqrt2\int_{\mathbb R^3}
\Psi_{1(-)}^\dagger\Psi_{2(-)}\,
\mathrm{d} u^1\mathrm{d} u^2\mathrm{d} u^3.
\label{eq2.29}
\end{equation}
On the space of arbitrary four-component functions, Eq.~\eqref{eq2.29} is not defined, since it does not contain $\Psi_{(+)}$. Under the conditions of Lemma~\ref{lem:projected-reduction}, the function $\Psi_{(+)}$ is uniquely determined by $\Psi_{(-)}$ through Eq.~\eqref{eq2.17}, and definition in Eq.~\eqref{eq2.29} gives a positive-definite scalar product on the space of functions $\Psi_{(-)}$.

Let us consider a moving matter with the velocity field,
\begin{equation}
v^1=\Omega u^2,
\qquad
v^2=v^3=0.
\label{eq2.30}
\end{equation}
The matter velocity is physically reasonable only if $|\bm v|<1$. Hence, the linear dependence in Eq.~\eqref{eq2.30} is directly applicable in the region $|\Omega u^2|<1$. The profile in Eq.~\eqref{eq2.30} describes a linear shear flow with the constant vorticity: $(\nabla\times\bm v)_3=-\Omega$. The symmetric part of the velocity gradient does not vanish: $D_{12}=D_{21}=\Omega/2$, where $D_{ij}=(\partial_{u^i}v^j+\partial_{u^j}v^i)/2$. Thus, Eq.~\eqref{eq2.30} does not describe rigid rotation of matter as a whole.

The Dirac equation for a neutrino in the velocity profile in Eq.~\eqref{eq2.30} was previously considered in Refs.~\cite{GrigorievSavochkinStudenikin2007,Studenikin2008}, where stationary solutions and a discrete spectrum were obtained in the massless approximation. In these works, the profile $v^1=\Omega u^2$, $v^2=v^3=0$ was interpreted as a model of a rotating medium. Its difference from the velocity field of rigid rotation $\bm v=\bm\Omega\times\bm r$ and the resulting difference between the spectra were discussed in Ref.~\cite{Dvornikov2011}. The Dirac equation for a neutrino in rigidly rotating matter, including in a comoving noninertial frame, was also studied in Refs.~\cite{Dvornikov2014,BreevDvornikov2026}.

For functions of the form given in Eq.~\eqref{eq2.21}, Eq.~\eqref{eq2.15} in the chiral representation of the gamma matrices gives
\begin{equation}
\Psi_{(+)}=
\begin{pmatrix}
-\dfrac{G_N\Omega u^2+\sqrt2\,\widehat{\widetilde p}_-}{\lambda-G_N}\,\psi_L
+\dfrac{\sqrt2\,m_\nu}{\lambda-G_N}\,\psi_R\\[3mm]
0\\0\\[1mm]
\dfrac{\sqrt2}{\lambda}\left(m_\nu\psi_L+\widehat{\widetilde p}_+\psi_R\right)
\end{pmatrix},
\qquad
\lambda\neq0,\quad \lambda\neq G_N.
\label{eq2.31}
\end{equation}
Substituting Eq.~\eqref{eq2.31} for $\Psi_{(+)}$ into Eq.~\eqref{eq2.14}, we obtain a closed system for the components $\psi_L$ and $\psi_R$:
\begin{align}
&\mathrm{i}\frac{\partial\Phi}{\partial u^0}=\widehat{\mathbb H}(m_\nu)\Phi,\quad
\widehat{\mathbb H}(m_\nu) =
\begin{pmatrix}
	\widehat{\mathbb H}_{11}&\widehat{\mathbb H}_{12}\\
	\widehat{\mathbb H}_{12}^{\dagger}&\widehat{\mathbb H}_{22}
\end{pmatrix},\quad
\Phi=\begin{pmatrix}\psi_L\\ \psi_R\end{pmatrix},\quad
\label{eq2.32}
\end{align}
\begin{equation}
\begin{aligned}
\widehat{\mathbb H}_{12}
&=-\frac{m_\nu G_N}{\lambda(\lambda-G_N)}
\left(\widehat{\widetilde p}_++\frac{\lambda\Omega u^2}{\sqrt2}\right),
\\[1mm]
\widehat{\mathbb H}_{11}
&=\frac{G_N}{2}
+\frac{(G_N\Omega u^2+\sqrt2\,\widehat{\widetilde p}_+)
(G_N\Omega u^2+\sqrt2\,\widehat{\widetilde p}_-)}{2(\lambda-G_N)}
+\frac{m_\nu^2}{\lambda},
\\[1mm]
\widehat{\mathbb H}_{22}
&=\frac{\widehat{\widetilde p}_1^2+\widehat{\widetilde p}_2^2}{\lambda}
+\frac{m_\nu^2}{\lambda-G_N},
\qquad \lambda\neq0,\quad \lambda\neq G_N.
\end{aligned}
\label{eq2.33}
\end{equation}
Thus, we have reduced the Dirac Eq.~(\ref{eq2.10}) to a Schr\"odinger-type equation with the Hamiltonian in Eq.~(\ref{eq2.33}).

\section{Spectrum of massless active neutrinos}\label{sec:oscillator-representation}

In this section, we consider massless neutrinos by setting $m_\nu=0$. In this situation, we derive the reduced Hamiltonian $\hat H$ describing active neutrinos in the velocity field in Eq.~\eqref{eq2.30}. Then, we express the Hamiltonian in terms of raising and lowering operators, derive the spectrum of massless active neutrinos, and construct a linear integral of motion $\hat A(u^0)$, which is used shortly in Sec.~\ref{sec:coherent-states} to construct time-dependent coherent solutions for neutrinos.

In the representation in Eq.~\eqref{eq2.25}, the right- and left-chiral spinor components decouple for $m_\nu=0$ and can be considered independently. Let us substitute the projected bispinor components in Eq.~\eqref{eq2.26} into the projected Eq.~\eqref{eq2.15}. For eigenfunctions of the symmetry operator $\widehat{\widetilde p}_3$, Eq.~\eqref{eq2.19} gives $2\widehat{\widetilde p}_3\Psi=\lambda\Psi$. The two nonzero components of Eq.~\eqref{eq2.15} then give the system,
\begin{equation}
	\begin{aligned}
		\bigl[\lambda-G_N(1-v^3)\bigr]\phi_L
		&=
		-\left[
		\sqrt2\bigl(\partial_{u^2}+\mathrm{i}\partial_{u^1}\bigr)
		+G_N\bigl(v^1-\mathrm{i}v^2\bigr)
		\right]\psi_L,
		\\
		\lambda\phi_R
		&=
		-\sqrt2\bigl(\partial_{u^2}-\mathrm{i}\partial_{u^1}\bigr)\psi_R .
	\end{aligned}
	\label{eq3.1}
\end{equation}
The other two components of Eq.~\eqref{eq2.15} are fulfilled identically. Thus, in the massless case, the equations for the spinors $(\psi_L,\phi_L)$ and $(\psi_R,\phi_R)$ are decoupled.

Active neutrino states satisfy the condition $(1+\gamma^5)\Psi/2=\Psi$. The explicit form of $\gamma^5$ in Eq.~\eqref{eq2.25} implies that $\psi_R=\phi_R=0$ for such states. Hence, the projected components of the bispinor describing active neutrinos have the form,
\begin{equation}
	\Psi_{(-)}^{(L)}
	=
	\begin{pmatrix}
		0\\
		\psi_L\\
		0\\
		0
	\end{pmatrix},
	\qquad
	\Psi_{(+)}^{(L)}
	=
	\begin{pmatrix}
		\phi_L\\
		0\\
		0\\
		0
	\end{pmatrix}.
	\label{eq3.2}
\end{equation}
For the velocity profile in Eq.~\eqref{eq2.30}, one gets that $v^1=\Omega u^2$, $v^2=v^3=0$. The first equation in Eq.~\eqref{eq3.1} takes the form,
\begin{equation}
	\begin{aligned}
		&(\lambda-G_N)\phi_L
		=
		-\left[
		\sqrt2\bigl(\partial_{u^2}+\mathrm{i}\partial_{u^1}\bigr)
		+G_N\Omega u^2
		\right]\psi_L,
		\\
		&\Psi_{(+)}^{(L)}
		=
		-\frac{1}{\lambda-G_N}
		\begin{pmatrix}
			\sqrt2\bigl(\partial_{u^2}+\mathrm{i}\partial_{u^1}\bigr)
			+G_N\Omega u^2\\
			0\\
			0\\
			0
		\end{pmatrix}\psi_L,
		\qquad
		\lambda\neq G_N.
	\end{aligned}
	\label{eq3.3}
\end{equation}
The condition $\lambda\neq0$, required to reconstruct $\phi_R$ from the second equation in Eq.~\eqref{eq3.1}, does not arise in the subsystem of active neutrino states, since $\psi_R=\phi_R=0$. Therefore, for $G_N\neq0$, the value $\lambda=0$ is admissible, although the reduction Eqs.~\eqref{eq2.16}--\eqref{eq2.17} is no longer meaningful.

Substituting $\Psi_{(+)}^{(L)}$ from Eq.~\eqref{eq3.3} into the first projected Eq.~\eqref{eq2.14}, we obtain the reduced Schr\"odinger equation,
\begin{equation}
	\mathrm{i}\partial_{u^0}\psi_L=\hat H\psi_L,
	\label{eq3.4}
\end{equation}
where the Hamiltonian $\hat H$ is given by
\begin{equation}
\hat H=
\frac{1}{G_N-\lambda}
\left[
\partial_{u^1}^2+\partial_{u^2}^2
-\mathrm{i}\sqrt2\,G_N\Omega\,u^2\partial_{u^1}
+\frac{G_N}{2}
\left(
G_N\left[1-\Omega^2(u^2)^2\right]
-\lambda+\sqrt2\,\Omega
\right)
\right].
\label{eq3.5}
\end{equation}
The solution of the original Dirac Eq.~\eqref{eq2.10} is represented as the sum of the projected components,
\begin{equation}
\Psi(u)=\Psi_{(-)}^{(L)}+\Psi_{(+)}^{(L)}=
\exp\!\left(-\frac{\mathrm{i}\lambda}{2}u^3\right)
\begin{pmatrix}
-\dfrac{\sqrt2(\partial_{u^2}+\mathrm{i}\partial_{u^1})+G_N\Omega u^2}
{\lambda-G_N}\\[2mm]
1\\0\\0
\end{pmatrix}\psi_L.
\label{eq3.6}
\end{equation}
Taking into account Eq.~\eqref{eq2.29}, the mean value of an operator $\widehat{\mathcal O}_\perp=\widehat{\mathcal O}_\perp(u^1,u^2)$ acting in the plane $(u^1,u^2)$ is given by
\begin{equation}
\langle \widehat{\mathcal O}_\perp\rangle
\equiv
\frac{\bigl(\Psi,\widehat{\mathcal O}_\perp\Psi\bigr)_{u^0}}
{\bigl(\Psi,\Psi\bigr)_{u^0}}
=\frac{\displaystyle\int_{\mathbb R^2}
\psi_L^*\widehat{\mathcal O}_\perp\psi_L\,\mathrm{d} u^1\mathrm{d} u^2}
{\displaystyle\int_{\mathbb R^2}|\psi_L|^2\,\mathrm{d} u^1\mathrm{d} u^2}.
\label{eq3.7}
\end{equation}

Let us compare the Hamiltonian $\hat H$, defined by~\eqref{eq3.5}, with the massive-neutrino Hamiltonian $\widehat{\mathbb H}(m_\nu)$ in~\eqref{eq2.33}:
\begin{equation}
\left.\widehat{\mathbb H}(m_\nu)\right|_{m_\nu=0} =
\begin{pmatrix}
\hat H&0\\[1mm]
0&\dfrac{\widehat{\widetilde p}_1^2+\widehat{\widetilde p}_2^2}{\lambda}
\end{pmatrix},\qquad
m_\nu=0,\quad \lambda\neq0,\quad \lambda\neq G_N.
\label{eq3.8}
\end{equation}
Thus, in the massless limit, the Hamiltonian in Eq.~\eqref{eq2.33} splits into the subsystem of active neutrinos interacting with matter, described by the Hamiltonian $\hat H$, and the free subsystem corresponding to sterile neutrinos.

It is convenient to introduce contravariant momentum components in the plane $(u^1,u^2)$,
\begin{equation}
\hat p^{\,j}\equiv\widehat{\widetilde p}^{\,j}
=-\widehat{\widetilde p}_j=-\mathrm{i}\partial_{u^j},
\qquad
[u^j,\hat p^{\,k}]=\mathrm{i}\delta^{jk},
\qquad j,k=1,2.
\label{eq3.9}
\end{equation}
In terms of the operators in Eq.~\eqref{eq3.9}, the Hamiltonian in Eq.~\eqref{eq3.5} takes the form,
\begin{equation}
\hat H=
\frac{1}{G_N-\lambda}
\left[
-(\hat p^{\,1})^2-(\hat p^{\,2})^2
+\sqrt2\,G_N\Omega\,u^2\hat p^{\,1}
-\frac{G_N^2\Omega^2}{2}(u^2)^2
+\frac{G_N}{2}\left(G_N-\lambda+\sqrt2\,\Omega\right)
\right].
\label{eq3.10}
\end{equation}
For $\Omega=0$ and $G_N\neq0$, motion in the plane $(u^1,u^2)$ is free,
\begin{equation}
\begin{aligned}
&\hat H_{\Omega=0}
=\frac{(\hat p^{\,1})^2+(\hat p^{\,2})^2}{\lambda-G_N}+\frac{G_N}{2},\\
&\psi_{\bm k}(u^0,u^1,u^2)=\frac{1}{2\pi}
\exp\left\{\mathrm{i} k_1u^1+\mathrm{i} k_2u^2
-\mathrm{i}\left[\frac{k_1^2+k_2^2}{\lambda-G_N}+\frac{G_N}{2}\right]u^0\right\}.
\end{aligned}
\label{eq3.11}
\end{equation}
The limit $G_N\to0$ will be considered separately in Sec.~\ref{subsec:free-limit}.
Completing the squaring  in Eq.~\eqref{eq3.10}, we finally obtain
\begin{align}
&\hat H=
-\frac{(\hat p^{\,1}-\beta u^2)^2+(\hat p^{\,2})^2}{G_N-\lambda}
+\frac{G_N\left(G_N-\lambda+\sqrt2\,\Omega\right)}{2(G_N-\lambda)},\label{eq3.12} \\
&\beta=\frac{G_N\Omega}{\sqrt2},\quad B=\left|\beta\right|, \quad s=\operatorname{sgn}\beta.
\label{eq3.13}
\end{align}
The Hamiltonian in Eq.~\eqref{eq3.12} admits an integral of motion
\begin{align}
	\hat{X} = \exp\left(\frac{\mathrm{i}\beta}{2} u^1 u^2\right)\hat{L}_3 \exp\left(-\frac{\mathrm{i}\beta}{2} u^1 u^2\right) = \mathrm{i} \left( u^2\partial_{u^1} -u^1\partial_{u^2} \right) -\frac{\beta}{2} \left[ (u^1)^2-(u^2)^2 \right],
	\label{eq3.14}
\end{align}
which is a unitary transform of the angular momentum operator along the $z$ axis, $\hat{L}_3 = x^1 \hat{p}^2 - x^2 \hat{p}^1$.

Let us introduce the kinetic momenta,
\begin{equation}
\hat\Pi_1=\hat p^{\,1}-\beta u^2,
\qquad
\hat\Pi_2=\hat p^{\,2},
\qquad
[\hat\Pi_1,\hat\Pi_2]=-\mathrm{i}\beta.
\label{eq3.15}
\end{equation}
We define the raising and lowering operators,
\begin{equation}
\begin{aligned}
&\hat a=\frac{\hat\Pi_1-\mathrm{i} s\hat\Pi_2}{\sqrt{2B}},
\qquad
\hat a^\dagger=\frac{\hat\Pi_1+\mathrm{i} s\hat\Pi_2}{\sqrt{2B}},\\
&[\hat a,\hat a^\dagger]=1,
\qquad
\hat\Pi_1^2+\hat\Pi_2^2=B\left(2\hat a^\dagger\hat a+1\right).
\end{aligned}
\label{eq3.16}
\end{equation}
Thus, in terms of these operators, the Hamiltonian in Eq.~\eqref{eq3.12} has the diagonal form,
\begin{equation}
\begin{aligned}
&\hat H=\varepsilon_0+\omega_a\,\hat a^\dagger\hat a,\\
&\omega_a=-\frac{2B}{G_N-\lambda}
=-\frac{\sqrt2\,|G_N\Omega|}{G_N-\lambda},\\
&\varepsilon_0=
\frac{G_N\left(G_N-\lambda+\sqrt2\,\Omega\right)-2B}{2(G_N-\lambda)}.
\end{aligned}
\label{eq3.17}
\end{equation}
In the basis $\vert n,\alpha\rangle$ of eigenfunctions of the operator
$\hat N_a=\hat a^\dagger\hat a$,
$\hat N_a\vert n,\alpha\rangle=n\vert n,\alpha\rangle$,
$n=0,1,2,\dots$, where $\alpha$ labels the degeneracy of a level, the spectrum of $\hat H$ is given by
\begin{equation}
\varepsilon_n(\lambda)=\varepsilon_0+\omega_a n
=
\frac{
\dfrac{G_N}{2}
\left(G_N-\lambda+\sqrt2\,\Omega\right)
-\dfrac{|G_N\Omega|}{\sqrt2}(2n+1)
}
{G_N-\lambda},
\qquad
n=0,1,2,\dots.
\label{eq3.18}
\end{equation}

The stationary Schr\"odinger equation associated with the reduced Eq.~\eqref{eq3.4} and the relation between the eigenvalue $\varepsilon_n(\lambda)$ and the physical energy read
\begin{equation}
\hat H\lvert n,\alpha\rangle
=\varepsilon_n(\lambda)\lvert n,\alpha\rangle,
\qquad
p_0=\frac{1}{\sqrt2}
\left[\varepsilon_n(\lambda)+\frac{\lambda}{2}\right].
\label{eq3.19}
\end{equation}
The degeneracy of each level $\varepsilon_n(\lambda)$ is resolved by choosing a joint basis of the operators $\hat H$ and $\hat X$. To this end, let us introduce a second pair of creation and annihilation operators $\hat b,\hat b^\dagger$, which commute with the first pair in Eq.~\eqref{eq3.16}:
\begin{equation}
\begin{aligned}
&\hat b=\frac{\hat K_1+\mathrm{i} s\hat K_2}{\sqrt{2B}},
\qquad
\hat b^\dagger=\frac{\hat K_1-\mathrm{i} s\hat K_2}{\sqrt{2B}},\qquad
\hat K_1=\hat p^{\,1},
\qquad
\hat K_2=\hat p^{\,2}-\beta u^1,\\
&[\hat b,\hat b^\dagger]=1,
\qquad
[\hat a,\hat b]=[\hat a,\hat b^\dagger]=[\hat a^\dagger,\hat b]=[\hat a^\dagger,\hat b^\dagger]=0.
\end{aligned}
\label{eq3.20}
\end{equation}
The coordinates $(u^1,u^2)$ and the momenta $(\hat{p}^1,\hat{p}^2)$ can then be expressed in terms of these operators as follows:
\begin{equation}
\begin{aligned}
&u^1=\frac{\mathrm{i}}{\sqrt{2B}}
\left(\hat a-\hat a^\dagger+\hat b-\hat b^\dagger\right),\quad
u^2=\frac{s}{\sqrt{2B}}
\left(\hat b+\hat b^\dagger-\hat a-\hat a^\dagger\right),\\
&\hat p^{\,1}=\sqrt{\frac{B}{2}}\left(\hat b+\hat b^\dagger\right),
\qquad
\hat p^{\,2}=\mathrm{i} s\sqrt{\frac{B}{2}}\left(\hat a-\hat a^\dagger\right),\quad
\hat{X}=s\left(\hat a^\dagger\hat a-\hat b^\dagger\hat b\right).
\end{aligned}
\label{eq3.21}
\end{equation}
Let us define the vacuum state, $\hat a\lvert 0,0\rangle=\hat b\lvert 0,0\rangle=0$, and introduce the basis
\begin{equation}
\lvert n,\ell\rangle=\frac{(\hat a^\dagger)^n(\hat b^\dagger)^\ell}{\sqrt{n!\,\ell!}}\lvert 0,0\rangle,\quad
\langle 0,0\vert 0,0\rangle = 1,
\quad
\hat X\lvert n,\ell\rangle=s(n-\ell)\lvert n,\ell\rangle,
\qquad n,\ell=0,1,2,\ldots.
\end{equation}
In the coordinate representation, this vacuum state is given by
\begin{equation}
\varphi_{00}(u^1,u^2)
=\sqrt{\frac{B}{2\pi}}
\exp\left[-\frac{B}{4}\left((u^1)^2+(u^2)^2\right)
+\frac{\mathrm{i} sB}{2}u^1u^2\right],
\qquad
\hat a\varphi_{00}=\hat b\varphi_{00}=0.
\label{eq3.23}
\end{equation}
Hence, the eigenfunctions of the commuting operators $\hat H$ and $\hat X$ in the coordinate representation, corresponding to the states $\lvert n,\ell\rangle$, have the form
\begin{equation}
\begin{aligned}
&\varphi_{n\ell}(u^1,u^2)
=\frac{(\hat a^\dagger)^n(\hat b^\dagger)^\ell}{\sqrt{n!\,\ell!}}\,
\varphi_{00}(u^1,u^2),\\
&\hat H\varphi_{n\ell}=\varepsilon_n(\lambda)\varphi_{n\ell},
\qquad
\hat X\varphi_{n\ell}=s(n-\ell)\varphi_{n\ell}.
\end{aligned}
\label{eq3.24}
\end{equation}
The eigenfunctions in Eq.~\eqref{eq3.24} are expressed in terms of two-dimensional complex Hermite polynomials,
\begin{equation}
\begin{aligned}
&\varphi_{n\ell}(u^1,u^2)
=\frac{\mathrm{i}^{\,n+\ell}}{\sqrt{n!\,\ell!}}\,
H_{n\ell}(w,\bar w)\,\varphi_{00}(u^1,u^2),\\
&H_{n\ell}(w,\bar w)
=\sum_{r=0}^{\min(n,\ell)}(-1)^r r!
\binom{n}{r}\binom{\ell}{r}
 w^{\,n-r}\bar w^{\,\ell-r},\\
&w=\sqrt{\frac{B}{2}}\left(u^1+\mathrm{i} s u^2\right),
\qquad
\bar w=\sqrt{\frac{B}{2}}\left(u^1-\mathrm{i} s u^2\right).
\end{aligned}
\label{eq3.25}
\end{equation}
The following orthogonality and completeness conditions hold, respectively:
\begin{equation}
\begin{aligned}
&\int_{\mathbb R^2}\varphi_{n\ell}^*(\bm u_\perp)
\varphi_{n'\ell'}(\bm u_\perp)\,\mathrm{d}^2u_\perp
=\delta_{nn'}\delta_{\ell\ell'},\\
&\sum_{n,\ell=0}^{\infty}
\varphi_{n\ell}(\bm u_\perp)
\varphi_{n\ell}^*(\bm u'_\perp)
=\delta^{(2)}(\bm u_\perp-\bm u'_\perp).
\end{aligned}
\label{eq3.26}
\end{equation}
where $\bm u_\perp=(u^1,u^2)$ and $\bm u'_\perp=(u'^1,u'^2)$.

Let us return to the Schr\"odinger-type Eq.~\eqref{eq3.4} for the function $\psi_L(u^0,u^1,u^2)$ and construct its linear integral of motion. The construction of time-dependent linear integrals of motion and the corresponding coherent states for a charged particle in a time-dependent magnetic field was developed in Ref.~\cite{MalkinMankoTrifonov1970}.

We look for an integral of motion in the form,
\begin{equation}
\hat A(u^0)=\mu(u^0)\hat a+\nu(u^0)\hat a^\dagger.
\label{eq3.27}
\end{equation}
Using the relations,
\begin{equation}
[\hat H,\hat a]=-\omega_a \hat a,
\qquad
[\hat H,\hat a^\dagger]=\omega_a \hat a^\dagger,\end{equation}
we obtain that
\begin{equation}
\begin{aligned}
&\dot\mu-\mathrm{i}\omega_a\mu=0,
\qquad
\dot\nu+\mathrm{i}\omega_a\nu=0,\\
&\mu(u^0)=\mu_0\mathrm{e}^{\mathrm{i}\omega_a u^0},
\qquad
\nu(u^0)=\nu_0\mathrm{e}^{-\mathrm{i}\omega_a u^0}.
\end{aligned}
\label{eq3.29}
\end{equation}
Here $\mu_0,\nu_0\in\mathbb C$. Let us require that the operators $\hat A(u^0)$ and $\hat A^\dagger(u^0)$ satisfy the standard oscillator commutation relation. Then
\begin{equation}
[\hat A(u^0),\hat A^\dagger(u^0)]=1,
\quad\Longleftrightarrow\quad
|\mu_0|^2-|\nu_0|^2=1.
\label{eq3.30}
\end{equation}
All solutions of the condition in Eq.~\eqref{eq3.30} can be parametrized by one complex parameter $\rho_0$ in the unit disk and a common phase $\chi$:
\begin{equation}
\mu_0=\frac{\mathrm{e}^{\mathrm{i}\chi}}{\sqrt{1-|\rho_0|^2}},
\qquad
\nu_0=\frac{\rho_0\mathrm{e}^{\mathrm{i}\chi}}{\sqrt{1-|\rho_0|^2}},
\qquad
|\rho_0|<1.
\label{eq3.31}
\end{equation}
The common phase $\mathrm{e}^{\mathrm{i}\chi}$ multiplies the operator $\hat A$ by a constant phase and can be absorbed into its complex eigenvalues. We therefore set $\chi=0$ without loss of generality. That is, $\mu_0=\cosh r$, $\nu_0=\mathrm{e}^{\mathrm{i}\vartheta}\sinh r$; then $\rho_0=\mathrm{e}^{\mathrm{i}\vartheta}\tanh r$, where $r\geq0$ and $\vartheta\in\mathbb R$. As a result, the time-dependent linear integral of motion is given by
\begin{equation}
\hat A(u^0)=
\frac{
\mathrm{e}^{\mathrm{i}\omega_a u^0}\hat a+
\rho_0\mathrm{e}^{-\mathrm{i}\omega_a u^0}\hat a^\dagger
}
{\sqrt{1-|\rho_0|^2}},\quad [\hat A(u^0),\hat A^\dagger(u^0)]=1,
\quad |\rho_0|<1.
\label{eq3.32}
\end{equation}

\section{Generalized coherent states}\label{sec:coherent-states}

In Sec.~\ref{subsec:joint-coherent}, we construct generalized coherent states $\psi_{z_1,z_2}$ of massless active neutrinos using the commuting operators $\hat A(u^0)$ and $\hat b$. In Sec.~\ref{subsec:complete-dirac-set}, we construct a complete set of solutions $\Psi_{z_1,z_2,\lambda}$ of the original Dirac Eq.~\eqref{eq2.10}. Then, in Sec.~\ref{subsec:mean-values}, we calculate the mean coordinates and momenta in the plane $(u^1,u^2)$ and obtain a closed system of Ehrenfest equations for these mean values.

The algebraic construction of generalized coherent states used below originates from the construction of coherent states of a charged particle in a uniform magnetic field~\cite{MalkinManko1969}. Coherent solutions of the Dirac equation for an electron on a null hyperplane were studied in Ref.~\cite{BagrovBuchbinderGitman1976}. Solutions with squeezed coherent states of transverse electron motion in light-cone variables were constructed in Refs.~\cite{GitmanBagrov2005,BagrovGavrilovGitmanMeira2011}. The general group formulation of generalized coherent states is presented in Ref.~\cite{Perelomov1986}. In this section, we apply this approach to the reduced Eq.~\eqref{eq3.4} for an active neutrino in moving matter.

\subsection{Coherent states of the operators \texorpdfstring{$\hat A(u^0)$ and $\hat b$}{A(u0) and b}}\label{subsec:joint-coherent}

Let us introduce the following definition:

\begin{definition}
\setlength{\emergencystretch}{2em}
A \emph{joint coherent state} of the operators $\hat A(u^0)$ and $\hat b$ is a solution $\psi_{z_1,z_2}(u^0,u^1,u^2)$ of Eq.~\eqref{eq3.4}, normalized in the Hilbert space $L^2(\mathbb R^2)$, that satisfies the system in Eq.~\eqref{eq4.1}
\begin{equation}
\begin{aligned}
&\hat A(u^0)\psi_{z_1,z_2}(u^0,u^1,u^2)
=z_1\psi_{z_1,z_2}(u^0,u^1,u^2),\\
&\hat b\psi_{z_1,z_2}(u^0,u^1,u^2)
=z_2\psi_{z_1,z_2}(u^0,u^1,u^2),
\qquad z_1,z_2\in\mathbb C.
\end{aligned}
\label{eq4.1}
\end{equation}
\end{definition}

The dependence on the fixed parameter $\lambda$ is not indicated in the notation $\psi_{z_1,z_2}$. In intermediate formulas, the arguments $u^0,u^1,u^2$ of this function are also omitted when this causes no ambiguity.

In the coordinate representation,
\begin{equation}
\begin{aligned}
&\hat a=\frac{-\mathrm{i}\partial_{u^1}-s\partial_{u^2}-\beta u^2}{\sqrt{2B}},
\qquad
\hat a^\dagger=\frac{-\mathrm{i}\partial_{u^1}+s\partial_{u^2}-\beta u^2}{\sqrt{2B}},\\
&\hat b=\frac{-\mathrm{i}\partial_{u^1}+s\partial_{u^2}-\mathrm{i} Bu^1}{\sqrt{2B}}.
\end{aligned}
\label{eq4.2}
\end{equation}
The parameter $\rho_0$, $|\rho_0|<1$, is defined by Eq.~\eqref{eq3.31}. Setting
\begin{equation}
\psi_{z_1,z_2}(u^0,u^1,u^2)=C_{z_1,z_2}(u^0)
\exp F(u^0,u^1,u^2),
\label{eq4.3}
\end{equation}
we obtain from Eq.~\eqref{eq4.1}
\begin{equation}
\begin{aligned}
&\partial_{u^1}F
=-\frac{B}{2}\left(1-\rho_0\mathrm{e}^{-2\mathrm{i}\omega_a u^0}\right)u^1
+\frac{\mathrm{i} sB}{2}\left(1+\rho_0\mathrm{e}^{-2\mathrm{i}\omega_a u^0}\right)u^2\\
&\hspace{2.2em}
+\frac{\mathrm{i}\sqrt{2B}}{2}\left[
\sqrt{1-|\rho_0|^2}\,z_1\mathrm{e}^{-\mathrm{i}\omega_a u^0}
+\left(1-\rho_0\mathrm{e}^{-2\mathrm{i}\omega_a u^0}\right)z_2\right],\\
&\partial_{u^2}F
=\frac{\mathrm{i} sB}{2}\left(1+\rho_0\mathrm{e}^{-2\mathrm{i}\omega_a u^0}\right)u^1
-\frac{B}{2}\left(1+\rho_0\mathrm{e}^{-2\mathrm{i}\omega_a u^0}\right)u^2\\
&\hspace{2.2em}
+\frac{s\sqrt{2B}}{2}\left[
\left(1+\rho_0\mathrm{e}^{-2\mathrm{i}\omega_a u^0}\right)z_2
-\sqrt{1-|\rho_0|^2}\,z_1\mathrm{e}^{-\mathrm{i}\omega_a u^0}\right].
\end{aligned}
\label{eq4.4}
\end{equation}
Integrating the system in Eq.~\eqref{eq4.4} with respect to $u^1$ and $u^2$, we find
\begin{equation}
\begin{aligned}
&F=-\frac{B}{4}\left[
\left(1-\rho_0\mathrm{e}^{-2\mathrm{i}\omega_a u^0}\right)(u^1)^2
+\left(1+\rho_0\mathrm{e}^{-2\mathrm{i}\omega_a u^0}\right)(u^2)^2\right]\\
&\quad+\frac{\mathrm{i} sB}{2}\left(1+\rho_0\mathrm{e}^{-2\mathrm{i}\omega_a u^0}\right)u^1u^2\\
&\quad+\frac{\mathrm{i}\sqrt{2B}}{2}\left[
\sqrt{1-|\rho_0|^2}\,z_1\mathrm{e}^{-\mathrm{i}\omega_a u^0}
+\left(1-\rho_0\mathrm{e}^{-2\mathrm{i}\omega_a u^0}\right)z_2\right]u^1\\
&\quad+\frac{s\sqrt{2B}}{2}\left[
\left(1+\rho_0\mathrm{e}^{-2\mathrm{i}\omega_a u^0}\right)z_2
-\sqrt{1-|\rho_0|^2}\,z_1\mathrm{e}^{-\mathrm{i}\omega_a u^0}\right]u^2.
\end{aligned}
\label{eq4.5}
\end{equation}
Hence,
\begin{equation}
\psi_{z_1,z_2}(u^0,u^1,u^2)
=C_{z_1,z_2}(u^0)\exp F(u^0,u^1,u^2).
\label{eq4.6}
\end{equation}

Substituting Eq.~\eqref{eq4.6} into the Schr\"odinger Eq.~\eqref{eq3.4}, we obtain
\begin{equation}
\mathrm{i}\frac{\dot C_{z_1,z_2}}{C_{z_1,z_2}}
=\varepsilon_0
+\omega_a\sqrt{1-|\rho_0|^2}\,z_1z_2\mathrm{e}^{-\mathrm{i}\omega_a u^0}
-\omega_a\rho_0z_2^2\mathrm{e}^{-2\mathrm{i}\omega_a u^0}.
\label{eq4.7}
\end{equation}
The integration gives
\begin{equation}
C_{z_1,z_2}(u^0)=C_{z_1,z_2}(0)
\exp\left\{
-\mathrm{i}\varepsilon_0u^0
+\sqrt{1-|\rho_0|^2}\,z_1z_2\left(\mathrm{e}^{-\mathrm{i}\omega_a u^0}-1\right)
-\frac{\rho_0z_2^2}{2}\left(\mathrm{e}^{-2\mathrm{i}\omega_a u^0}-1\right)
\right\}.
\label{eq4.8}
\end{equation}
Normalization determines the modulus of the constant $C_{z_1,z_2}(0)$ rather than its overall phase. Let us require
\begin{equation}
\int_{\mathbb R^2}|\psi_{z_1,z_2}(0,u^1,u^2)|^2\,\mathrm{d} u^1\mathrm{d} u^2=1.
\label{eq4.9}
\end{equation}
We choose the remaining phase so that the normalization constant has the compact form,
\begin{equation}
\begin{aligned}
&C_{z_1,z_2}(0)=\sqrt{\frac{B\sqrt{1-|\rho_0|^2}}{2\pi}}\\
&\quad\times\exp\left[
-\frac{|z_1|^2+|z_2|^2}{2}
+\frac{\rho_0^*z_1^2}{2}
+\sqrt{1-|\rho_0|^2}\,z_1z_2
-\frac{\rho_0z_2^2}{2}
\right].
\end{aligned}
\label{eq4.10}
\end{equation}
The absolute value of Eq.~\eqref{eq4.10} is determined solely by the condition in Eq.~\eqref{eq4.9}. The chosen phase affects neither the absolute values of the scalar products nor the mean values.

Taking into account Eq.~\eqref{eq3.23} for the vacuum state $\varphi_{00}$, together with Eqs.~\eqref{eq4.8}, \eqref{eq4.10}, and~\eqref{eq4.5}, we finally obtain that
\begin{equation}
\begin{aligned}
&\psi_{z_1,z_2}(u^0,u^1,u^2)
=(1-|\rho_0|^2)^{1/4}\,\varphi_{00}(u^1,u^2)\\
&\quad\times\exp\Bigg\{
-\frac{|z_1|^2+|z_2|^2}{2}
+\frac{\rho_0^*z_1^2}{2}
-\mathrm{i}\varepsilon_0u^0\\
&\qquad
+\frac{\rho_0}{2}\mathrm{e}^{-2\mathrm{i}\omega_a u^0}(w-\mathrm{i} z_2)^2
+\mathrm{i}\sqrt{1-|\rho_0|^2}\,z_1\mathrm{e}^{-\mathrm{i}\omega_a u^0}(w-\mathrm{i} z_2)
 +\mathrm{i} z_2\bar w
\Bigg\},\\
&\qquad
w=\sqrt{\frac{B}{2}}\,(u^1+\mathrm{i} s u^2),
\qquad
\bar w=\sqrt{\frac{B}{2}}\,(u^1-\mathrm{i} s u^2).
\end{aligned}
\label{eq4.11}
\end{equation}
Here $z_1$ and $z_2$ are the eigenvalues of the operators $\hat A(u^0)$ and $\hat b$. The complex parameter $\rho_0$, having $|\rho_0|<1$, specifies the linear integral of motion in Eq.~\eqref{eq3.32}. The quantities $\varepsilon_0$ and $\omega_a$ are defined by Eq.~\eqref{eq3.17}. The explicit dependence of $\psi_{z_1,z_2}$ on $u^1,u^2$, before it is expressed in terms of $w,\bar w$, is given in Eq.~\eqref{eq4.5}.

\subsection{A complete set of solutions of the original Dirac equation}\label{subsec:complete-dirac-set}

For a fixed $\rho_0$, $|\rho_0|<1$, and each $\lambda\neq G_N$, let us substitute the function in Eq.~\eqref{eq4.11} into Eq.~\eqref{eq3.6}. The full solution of the original Dirac Eq.~\eqref{eq2.10} takes the form,
\begin{equation}
\Psi_{z_1,z_2,\lambda}(u)
=\frac{\mathrm{e}^{-\mathrm{i}\lambda u^3/2}}{\sqrt{4\pi\sqrt2}}
\begin{pmatrix}
-\dfrac{\sqrt2(\partial_{u^2}+\mathrm{i}\partial_{u^1})+G_N\Omega u^2}
{\lambda-G_N}\\[2mm]
1\\0\\0
\end{pmatrix}
\psi_{z_1,z_2}(u^0,u^1,u^2),
\qquad z_1,z_2\in\mathbb C.
\label{eq4.12}
\end{equation}

The following lemma holds:

\begin{lemma}\label{lem:coherent-dirac-completeness}
For a fixed $\rho_0$, $|\rho_0|<1$, the functions $\Psi_{z_1,z_2,\lambda}$, where $z_1,z_2\in\mathbb C$ and $\lambda\in\mathbb R\setminus\{G_N\}$, are solutions of the original Dirac Eq.~\eqref{eq2.10}, normalized with respect to the longitudinal spectral parameter,
\begin{equation}
\left(\Psi_{z_1,z_2,\lambda},\Psi_{z_1,z_2,\lambda'}\right)_{u^0}
=\delta(\lambda-\lambda').
\label{eq4.13}
\end{equation}
At fixed $u^0$ and $\lambda$, the coherent states are not orthogonal. The absolute value of the scalar product of their functions $\psi_{z_1,z_2}$ in $L^2(\mathbb R^2)$ is
\begin{equation}
\left|\int_{\mathbb R^2}
\psi_{z_1,z_2}^*(u^0,u^1,u^2)
\psi_{z'_1,z'_2}(u^0,u^1,u^2)\,\mathrm{d} u^1\mathrm{d} u^2\right|
=\exp\left[-\frac{|z_1-z'_1|^2+|z_2-z'_2|^2}{2}\right].
\label{eq4.14}
\end{equation}
The resolution of the identity has the form,
\begin{equation}
\int_{\mathbb R\setminus\{G_N\}}\mathrm{d}\lambda
\int_{\mathbb C}\frac{\mathrm{d}^2z_1}{\pi}
\int_{\mathbb C}\frac{\mathrm{d}^2z_2}{\pi}\,
\lvert \Psi_{z_1,z_2,\lambda}\rangle
\langle \Psi_{z_1,z_2,\lambda}\rvert
=I_{\mathcal H_L}.
\label{eq4.15}
\end{equation}
Here $\mathcal H_L$ is the Hilbert space of the active neutrino states under consideration. On the null hypersurface $u^0=\mathrm{const}$, it is realized as the space of independent projected components $\Psi_{(-)}$ with the scalar product in Eq.~\eqref{eq2.29}.
\end{lemma}

\begin{proof}
At each fixed $u^0$, the operator $\hat A(u^0)$  in Eq.~\eqref{eq3.32} satisfies $[\hat A(u^0),\hat A^\dagger(u^0)]=1$, the operator $\hat b$ satisfies~\eqref{eq3.20}, and all other commutators vanish. Hence, $\psi_{z_1,z_2}(u^0,\cdot,\cdot)$ is a normalized joint coherent state of two independent sets of creation and annihilation operators.

At fixed \(u^0\), the transformation in Eq.~\eqref{eq3.32} is a canonical Bogoliubov transformation. If \(|z_1,z_2\rangle\) are the standard joint coherent states of the operators \(\hat a\) and \(\hat b\), then
$$
|z_1,z_2\rangle = 
\mathrm{e}^{-(|z_1|^2+|z_2|^2)/2}
\sum_{n,\ell=0}^{\infty}
\frac{z_1^nz_2^\ell}{\sqrt{n!\ell!}}
|n,\ell\rangle.
$$
Taking into account the integral
$$
\int_{\mathbb C}\frac{\mathrm d^2z}{\pi}
\mathrm{e}^{-|z|^2}z^n(z^*)^m = 
n!\delta_{nm}
$$
and the completeness of the basis in Eq.~\eqref{eq3.26}, we obtain the standard resolution of the identity. Since the states \(\psi_{z_1,z_2}(u^0)\) are obtained from these states by the same unitary transformation, at fixed $\lambda$ we have
\begin{equation}
\int_{\mathbb C}\frac{\mathrm d^2z_1}{\pi}
\int_{\mathbb C}\frac{\mathrm d^2z_2}{\pi}
\lvert\psi_{z_1,z_2}(u^0)\rangle
\langle\psi_{z_1,z_2}(u^0)\rvert
=I_{L^2(\mathbb R^2)}.
\label{eq4.16}
\end{equation}
The standard scalar product of two coherent states implies~\eqref{eq4.14}; in particular, distinct coherent states are generally not orthogonal.

According to Eq.~\eqref{eq2.29}, the scalar product of the full bispinors depends only on the component $\Psi_{(-)}^{(L)}$, which contains the function $\psi_{z_1,z_2}$. Therefore, the normalization of $\psi_{z_1,z_2}$ in $L^2(\mathbb R^2)$ and the longitudinal Fourier relation in Eq.~\eqref{eq2.22}, together with the factor $(4\pi\sqrt2)^{-1/2}$ in Eq.~\eqref{eq4.12}, give Eq.~\eqref{eq4.13}. Combining the longitudinal Fourier transform with the transverse resolution of the identity in Eq.~\eqref{eq4.16}, it proves Eq.~\eqref{eq4.15}.
\end{proof}

\subsection{Mean coordinates and momenta in the plane \texorpdfstring{$(u^1,u^2)$}{(u1,u2)}}\label{subsec:mean-values}

From Eq.~\eqref{eq3.32}, we obtain the inverse transformation,
\begin{equation}
\hat a=
\frac{\mathrm{e}^{-\mathrm{i}\omega_a u^0}}
{\sqrt{1-|\rho_0|^2}}
\left(\hat A-\rho_0\hat A^\dagger\right).
\label{eq4.17}
\end{equation}
Hence, for the state $\psi_{z_1,z_2}(u^0,\cdot,\cdot)$,
\begin{equation}
\langle \hat a\rangle=
\frac{z_1-\rho_0z_1^*}{\sqrt{1-|\rho_0|^2}}\,
\mathrm{e}^{-\mathrm{i}\omega_a u^0},
\qquad
\langle \hat b\rangle=z_2.
\label{eq4.18}
\end{equation}
Using the inverse Eq.~\eqref{eq3.21}, we obtain
\begin{equation}
\begin{aligned}
&\langle u^1\rangle+\mathrm{i} s\langle u^2\rangle
=\mathrm{i}\sqrt{\frac{2}{B}}\left[
 z_2-
\frac{z_1^*-\rho_0^*z_1}{\sqrt{1-|\rho_0|^2}}\,
\mathrm{e}^{\mathrm{i}\omega_a u^0}
\right],\\
&\langle \hat p^{\,1}\rangle=\sqrt{2B}\,\operatorname{Re} z_2,\qquad
\langle \hat p^{\,2}\rangle=-s\sqrt{\frac{2B}{1-|\rho_0|^2}}\,
\operatorname{Im}\!\left[(z_1-\rho_0z_1^*)\mathrm{e}^{-\mathrm{i}\omega_a u^0}\right].
\end{aligned}
\label{eq4.19}
\end{equation}
The real and imaginary parts of the first equality  in Eq.~\eqref{eq4.19} give, respectively, $\langle u^1\rangle$ and
$s\langle u^2\rangle$.
For the coordinate $u^2$, on which the velocity field in Eq.~\eqref{eq2.30} depends, the variance is
\begin{equation}
(\Delta u^2)^2
=
\frac{1-\operatorname{Re}\!\left(\rho_0\mathrm{e}^{-2\mathrm{i}\omega_a u^0}\right)}
{B\left(1-|\rho_0|^2\right)}.
\label{eq4.20}
\end{equation}

The velocity field in Eq.~\eqref{eq2.30} has a physical meaning in the region
\[
|\Omega u^2|<1.
\]
A coherent state has a nonzero width along \(u^2\); hence, the condition for the applicability of the velocity profile should be formulated for the region in which this coherent state is concentrated. Let us denote \(\mu_2=\langle u^2\rangle\), \(\sigma_2=\Delta u^2\) and consider the region
\[
|u^2-\mu_2|\leq\kappa\sigma_2,
\qquad \kappa>0.
\]
Within this region, one has that
\[
|u^2|\leq|\mu_2|+\kappa\sigma_2.
\]
Thus, a sufficient condition for this region to lie entirely within \(|\Omega u^2|<1\) is
\begin{equation}
|\Omega|
\left(
|\langle u^2\rangle|
+\kappa\Delta u^2
\right)<1.
\label{eq4.21}
\end{equation}
Since the \(u^2\) distribution for the coherent states under consideration is Gaussian, under this condition the probability of leaving the region of applicability of the velocity profile does not exceed the quantity
\[
\operatorname{erfc}\left(\frac{\kappa}{\sqrt2}\right).
\]

According to Eq.~\eqref{eq2.29}, the mean values of the operators $u^1$, $u^2$, $\hat p^{\,1}$, and $\hat p^{\,2}$ calculated using $\psi_{z_1,z_2}$ coincide with the mean values of the same operators for the bispinor $\Psi_{z_1,z_2,\lambda}$  in Eq.~\eqref{eq4.12}. Recall that, at fixed $\lambda$, the bispinor $\Psi_{z_1,z_2,\lambda}$ is a generalized eigenstate of the operator $\widehat{\widetilde p}_3$:
\begin{equation}
\widehat{\widetilde p}_3\Psi_{z_1,z_2,\lambda}
=\frac{\lambda}{2}\Psi_{z_1,z_2,\lambda}.
\label{eq4.22}
\end{equation}

Using the Hamiltonian in Eq.~\eqref{eq3.12} and the definitions of the kinetic momenta in Eq.~\eqref{eq3.15}, we obtain a closed system of Ehrenfest equations for the mean values:
\begin{equation}
\begin{aligned}
&\frac{\mathrm{d}\langle u^1\rangle}{\mathrm{d} u^0}
=-\frac{2}{G_N-\lambda}\langle \hat\Pi_1\rangle,
\qquad
\frac{\mathrm{d}\langle u^2\rangle}{\mathrm{d} u^0}
=-\frac{2}{G_N-\lambda}\langle \hat\Pi_2\rangle,\\
&\frac{\mathrm{d}\langle \hat\Pi_1\rangle}{\mathrm{d} u^0}
=\frac{2\beta}{G_N-\lambda}\langle \hat\Pi_2\rangle,
\qquad
\frac{\mathrm{d}\langle \hat\Pi_2\rangle}{\mathrm{d} u^0}
=-\frac{2\beta}{G_N-\lambda}\langle \hat\Pi_1\rangle.
\end{aligned}
\label{eq4.23}
\end{equation}
In particular,
\begin{equation}
\frac{\mathrm{d}^2\langle \hat\Pi_j\rangle}{\mathrm{d} (u^0)^2}
+\omega_a^2\langle \hat\Pi_j\rangle=0,
\qquad j=1,2.
\label{eq4.24}
\end{equation}

\section{Classical interpretation and the transition to free motion}\label{sec:classical-limits}

In this section, we first clarify the relation between the evolution parameter $u^0$ and the Minkowski coordinates. We then compare the system of Ehrenfest Eqs.~\eqref{eq4.23} with a classical Hamiltonian system in the plane $(u^1,u^2)$. The limit $G_N\to0$ is studied in Sec.~\ref{subsec:free-limit}.

The transformations between $x^\mu$ and $u^\mu$ are given by~\eqref{eq2.4}. In particular,
\begin{equation}
u^0=x^-=\frac{x^0-x^3}{\sqrt2}.
\label{eq5.1}
\end{equation}
The variable $u^0=x^-$ is the evolution parameter in the Schr\"odinger Eq.~\eqref{eq3.4} and in the system of Ehrenfest equations~\eqref{eq4.23}. Since solutions with fixed $\lambda$ are plane waves in $u^3=x^+$, the mean value $\langle u^3\rangle$ is not defined for them. Therefore, these states do not determine the full mean worldline $\langle x^\mu\rangle$. In particular, the mean values $\langle u^1\rangle$ and $\langle u^2\rangle$ cannot be uniquely reparametrized by $x^0$ without additional localization of the wave packet in the direction perpendicular to the plane $(u^1,u^2)$.

\subsection{A classical system in the plane \texorpdfstring{$(u^1,u^2)$}{(u1,u2)}}\label{subsec:classical-system}

Let us write the classical symbol of the Hamiltonian in Eq.~\eqref{eq3.12}, with the constant term omitted, as follows. We denote the classical symbols of the operators $\hat p^{\,1}$ and $\hat p^{\,2}$ by $p^1$ and $p^2$:
\begin{equation}
H_{\mathrm{cl}}
=\frac{(p^1-\beta u^2)^2+(p^2)^2}{2m_{\mathrm{eff}}},
\qquad
m_{\mathrm{eff}}=\frac{\lambda-G_N}{2}.
\label{eq5.2}
\end{equation}
Here $m_{\mathrm{eff}}$ is a parameter of the reduced dynamics in the plane $(u^1,u^2)$ with the evolution parameter $u^0$. It differs from the neutrino mass $m_\nu$. The sign of $m_{\mathrm{eff}}$, which is determined by that of $\lambda-G_N$, can be either positive or negative. Therefore, the analogy with charged-particle motion formulated below Eq.~\eqref{eq5.4} is to be understood as a mathematical equivalence of the Hamiltonian systems in the plane $(u^1,u^2)$. The Lagrangian obtained from $H_{\mathrm{cl}}$  in Eq.~\eqref{eq5.2} by the Legendre transformation is given by
\begin{equation}
L=\frac{m_{\mathrm{eff}}}{2}
\left[(\dot u^1)^2+(\dot u^2)^2\right]
+\beta u^2\dot u^1,
\qquad
\dot u^j=\frac{\mathrm{d} u^j}{\mathrm{d} u^0}.
\label{eq5.3}
\end{equation}
A Lagrangian, which differs from Eq.~\eqref{eq5.3} by a total derivative, can be written in the symmetric form,
\begin{equation}
L_{\mathrm{sym}}
=\frac{m_{\mathrm{eff}}}{2}
\left[(\dot u^1)^2+(\dot u^2)^2\right]
+\frac{\beta}{2}\left(u^2\dot u^1-u^1\dot u^2\right).
\label{eq5.4}
\end{equation}
The Hamiltonian system generated by $H_{\mathrm{cl}}$  in Eq.~\eqref{eq5.2} is mathematically equivalent to the motion of a fictitious charged particle in the plane $(u^1,u^2)$ in a uniform magnetic field perpendicular to this plane. The quantity $(-\beta)$ plays the role of the product of the charge and the magnetic field. Hence, the cyclotron frequency with respect to $u^0$ is
\begin{equation}
\omega_c=-\frac{\beta}{m_{\mathrm{eff}}}
=\frac{2\beta}{G_N-\lambda}
=\frac{\sqrt2\,G_N\Omega}{G_N-\lambda}.
\label{eq5.5}
\end{equation}
Taking into account Eq.~\eqref{eq3.17}, we also have $\omega_c=-s\omega_a$. The center of the orbit is determined by
\begin{equation}
u_c^1=u^1-\frac{p^2}{\beta},
\qquad
u_c^2=\frac{p^1}{\beta},
\label{eq5.6}
\end{equation}
and the classical trajectory in the plane $(u^1,u^2)$ has the form
\begin{equation}
\begin{aligned}
&u^1(u^0)=u_c^1+R\cos(\omega_cu^0+\varphi),\\
&u^2(u^0)=u_c^2-R\sin(\omega_cu^0+\varphi),\\
&(u^1-u_c^1)^2+(u^2-u_c^2)^2=R^2.
\end{aligned}
\label{eq5.7}
\end{equation}
The orbit radius is determined by the initial condition,
\begin{equation}
R^2=\frac{(p^2_0)^2+(p^1_0-\beta u^2_0)^2}{\beta^2},
\qquad
\beta=\frac{G_N\Omega}{\sqrt2}.
\label{eq5.8}
\end{equation}
Here $p^j_0$ and $u^2_0$ are the values of $p^j$ and $u^2$ at $u^0=0$. The Hamilton equations, resulting from Eq.~\eqref{eq5.2}, coincide with the Ehrenfest Eqs.~\eqref{eq4.23} if we identify $u^j\leftrightarrow\langle u^j\rangle$ and
$p^j\leftrightarrow\langle \hat p^{\,j}\rangle$, $j=1,2$.

\subsection{The limit \texorpdfstring{$B\to0$}{B -> 0} and the transition to free motion}
\label{subsec:free-limit}

In this section, we examine the constructed coherent states in the limit $B \to 0$, when the oscillatory structure of the reduced Hamiltonian is violated. Since the operators $\hat a$ and $\hat b$ contain a factor $B^{-1/2}$, the average transverse momentum of a coherent state approaches to zero for fixed $z_1,z_2$. Thus, an explicit limit to the free motion is ambiguous. 

For this purpose, we, first, consider a $B$-dependent family of coherent states, with the parameters being chosen such that the initial average coordinates and momenta remain fixed. Then, the limit $B\to0$ is implemented by considering 
$G_N\to0$ for fixed $\lambda\neq0$ and $\Omega\neq0$. In this
case, the neutrino interaction with matter also vanishes. We will show that after the appropriate scaling, the coherent functions transform into the plane waves of a free massless neutrino, as well as the motion of the wave packet center occurs along a free classical trajectory.

Let us fix the spectral parameter $\lambda\neq0$ and the parameter $\Omega\neq0$. We consider the limit $G_N\to0^+$. For $G_N>0$, the value $s=\operatorname{sgn}\Omega$ remains constant, and
\begin{equation}
	\begin{aligned}
		&\beta\to0,
		\qquad
		B\to0,\\
		&\hat H_{\mathrm{free}} = \lim_{G_N\to0^+} \hat H =
		-\frac{1}{\lambda}
		\left(
		\partial_{u^1}^2+\partial_{u^2}^2
		\right)
		=
		\frac{(\hat p^{\,1})^2+(\hat p^{\,2})^2}{\lambda}.
	\end{aligned}
	\label{eq5.9}
\end{equation}
Plane-wave solutions of the reduced Eq.~\eqref{eq3.4} in the free limit, with the Hamiltonian in Eq.~\eqref{eq5.9}, have the form,
\begin{equation}
	\psi_{\bm k}^{(0)}(u^0,u^1,u^2)
	=
	\frac{1}{2\pi}
	\exp\left[
	\mathrm{i}k_1u^1+\mathrm{i}k_2u^2
	-\mathrm{i}\frac{k_1^2+k_2^2}{\lambda}u^0
	\right].
	\label{eq5.10}
\end{equation}

The operators $\hat a$ and $\hat b$ contain the factor $B^{-1/2}$ and have no finite operator limit at $B\to0$. At fixed $z_1,z_2$ and fixed $\rho_0$, $|\rho_0|<1$, Eq.~\eqref{eq4.19} implies $\langle\hat p^{\,1}\rangle=O(\sqrt B)$ and $\langle\hat p^{\,2}\rangle=O(\sqrt B)$. Hence, the mean transverse momentum tends to zero.

For \(B>0\), the parameters \(z_1,z_2\) are eigenvalues of the operators \(\hat A(u^0)\) and \(\hat b\). To select a family of coherent states that tends, as \(G_N\to0\), to the free plane wave in Eq.~\eqref{eq5.10} with a given \(\bm k\), let us relate the parameters \(z_1,z_2\) to the initial mean coordinates and momenta:
\begin{equation}
	u_0^j=\langle u^j\rangle(0),
	\qquad
	k_j=\langle\hat p^{\,j}\rangle(0),
	\qquad
	j=1,2.
	\label{eq5.11}
\end{equation}
Eq.~\eqref{eq4.19}, at $u^0=0$ we have
\begin{equation}
	u_0^1+\mathrm{i}s u_0^2
	=
	\mathrm{i}\sqrt{\frac{2}{B}}
	\left[
	z_2-
	\frac{z_1^*-\rho_0^*z_1}
	{\sqrt{1-|\rho_0|^2}}
	\right],
	\label{eq5.12}
\end{equation}
\begin{equation}
	k_1
	=
	\sqrt{2B}\,\operatorname{Re}z_2,
	\qquad
	k_2
	=
	-s\sqrt{\frac{2B}{1-|\rho_0|^2}}\,
	\operatorname{Im}\left(z_1-\rho_0z_1^*\right).
	\label{eq5.13}
\end{equation}
Solving these relations for $z_1,z_2$, we obtain, for any $B>0$,
\begin{align}
	&z_2 = z_2(G_N) =
	\frac{k_1+\mathrm{i}s k_2}{\sqrt{2B}}
	-\mathrm{i}\sqrt{\frac{B}{2}}\,u_0^1,
	\nonumber\\
	&z_1= z_1(G_N) = 
	\frac{(1+\rho_0)k_1-\mathrm{i}s(1-\rho_0)k_2}
	{\sqrt{2B(1-|\rho_0|^2)}}
	-s\sqrt{\frac{B}{2}}\,
	\frac{1+\rho_0}{\sqrt{1-|\rho_0|^2}}\,
	u_0^2.
	\label{eq5.14}
\end{align}
Thus, the free limit should be considered for a family of coherent states with fixed $u_0^j$ and $k_j$. Hence,
\begin{equation}
	z_1=O(B^{-1/2})=O(|G_N|^{-1/2}),
	\qquad
	z_2=O(B^{-1/2})=O(|G_N|^{-1/2}).
	\label{eq5.15}
\end{equation}
In Eq.~\eqref{eq4.11}, the quadratic dependence on $u^1,u^2$ is contained in the exponent of $\varphi_{00}$ and in the term $\rho_0\mathrm{e}^{-2\mathrm{i}\omega_a u^0}w^2/2$, where
$w=\sqrt{B/2}\,(u^1+\mathrm{i}s u^2)$.
Therefore, the entire part of the exponent of \(\psi_{z_1,z_2}\) that is quadratic in \(u^1,u^2\) is of order \(B\) and tends to zero as $B\to0$.

For each $B>0$, the function $\psi_{z_1,z_2}$ is normalized in $L^2(\mathbb R^2)$. At $B\to0$, the width of the corresponding Gaussian packet increases without bound, whereas the free plane wave in Eq.~\eqref{eq5.10} does not belong to $L^2(\mathbb R^2)$ and is delta-normalized with respect to $\bm k$. To obtain the plane-wave limit, let us introduce a function $\widetilde\psi_{\bm k,B}$ differing from $\psi_{z_1,z_2}$ by a factor $\mathcal R_B$ that depends on $B$ and is independent of the coordinates. We choose $\mathcal R_B$ so that, at \(u^0=u^1=u^2=0\), the value of \(\widetilde\psi_{\bm k,B}\) coincides with that of the free plane wave in Eq.~\eqref{eq5.10}, namely \(1/(2\pi)\). From Eqs.~\eqref{eq4.11} and~\eqref{eq5.14}, we obtain
\begin{subequations}
	\label{eq5.16}
	\begin{flalign}
		&\widetilde\psi_{\bm k,B}
		=
		\mathcal R_B\psi_{z_1,z_2},
		\qquad
		\widetilde\psi_{\bm k,B}(0,0,0)
		=
		\frac{1}{2\pi},
		&&
		\label{eq5.16a}
		\\[1mm]
		&\mathcal R_B
		=
		\frac{1}{2\pi\,\psi_{z_1,z_2}(0,0,0)}
		\nonumber
		\\
		&\hspace{2em}
		=
		\frac{1}
		{\sqrt{2\pi B}\,(1-|\rho_0|^2)^{1/4}}
		\exp\left[
		\frac{|z_1|^2+|z_2|^2}{2}
		-\frac{\rho_0^*z_1^2}{2}
		-\sqrt{1-|\rho_0|^2}\,z_1z_2
		+\frac{\rho_0z_2^2}{2}
		\right],
		&&
		\label{eq5.16b}
		\\[1mm]
		&\sqrt{1-|\rho_0|^2}\,z_1z_2-\rho_0z_2^2
		=
		\left[
		\sqrt{1-|\rho_0|^2}\,z_1-\rho_0z_2
		\right]z_2
		=
		\frac{k_1^2+k_2^2}{2B}+O(1),
		\nonumber
		\\
		&
		\omega_a
		=
		\frac{2B}{\lambda}+O(B^2),
		\qquad
		\varepsilon_0=O(B),
		&&
		\label{eq5.16c}
		\\[1mm]
		&\mathrm{e}^{-\mathrm{i}\omega_a u^0}-1
		=
		-\mathrm{i}\omega_a u^0+O(B^2),
		\qquad
		\mathrm{e}^{-2\mathrm{i}\omega_a u^0}-1
		=
		-2\mathrm{i}\omega_a u^0+O(B^2),
		\nonumber
		\\
		&-\mathrm{i}\varepsilon_0u^0
		+\sqrt{1-|\rho_0|^2}\,z_1z_2
		\left(
		\mathrm{e}^{-\mathrm{i}\omega_a u^0}-1
		\right)
		-\frac{\rho_0z_2^2}{2}
		\left(
		\mathrm{e}^{-2\mathrm{i}\omega_a u^0}-1
		\right)
		=
		-\mathrm{i}\frac{k_1^2+k_2^2}{\lambda}u^0+O(B).
		&&
		\label{eq5.16d}
	\end{flalign}
\end{subequations}
In Eq.~\eqref{eq5.16d}, $u^0$ is fixed. According to Eq.~\eqref{eq5.14}, $z_1z_2=O(B^{-1})$ and $z_2^2=O(B^{-1})$. Hence, after multiplication by $z_1z_2$ or $z_2^2$, the $O(B^2)$ remainders in the expansions of the exponentials give a contribution
$O(B)$.

Substituting Eq.~\eqref{eq5.14} into the part of the exponent  in Eq.~\eqref{eq4.11} that is linear in $u^1$ and $u^2$, gives $\mathrm{i}k_1u^1+\mathrm{i}k_2u^2+O(B)$. The quadratic part of the exponent in $u^1,u^2$ is of order $B$. Thus, at fixed $u^0$, the function $\widetilde\psi_{\bm k,B}$ tends to the free plane wave. In particular, in the sense of distributions,
\begin{equation}
	\widetilde\psi_{\bm k,B}(u^0,u^1,u^2)
	\longrightarrow
	\psi_{\bm k}^{(0)}(u^0,u^1,u^2)
	=
	\frac{1}{2\pi}
	\exp\left[
	\mathrm{i}k_1u^1+\mathrm{i}k_2u^2
	-\mathrm{i}\frac{k_1^2+k_2^2}{\lambda}u^0
	\right].
	\label{eq5.17}
\end{equation}

The corresponding solution of the free Dirac equation obtained from Eq.~\eqref{eq2.10} for $G_N=0$ and $m_\nu=0$ has the form,
\begin{equation}
	\Psi_{\bm k}^{(0)}(u)
	=
	\frac{1}{2\pi\sqrt{4\pi\sqrt2}}
	\exp\left[
	-\frac{\mathrm{i}\lambda}{2}u^3
	+\mathrm{i}k_1u^1+\mathrm{i}k_2u^2
	-\mathrm{i}\frac{k_1^2+k_2^2}{\lambda}u^0
	\right]
	\begin{pmatrix}
		\dfrac{\sqrt2}{\lambda}(k_1-\mathrm{i}k_2)\\[2mm]
		1\\
		0\\
		0
	\end{pmatrix}.
	\label{eq5.18}
\end{equation}
Direct substitution shows that Eq.~\eqref{eq5.18} satisfies the free massless Dirac equation obtained from Eq.~\eqref{eq2.10} for $G_N=0$ and $m_\nu=0$.

Let us separately consider the free Ehrenfest system for a localized wave packet with initial mean coordinates \(u_0^j\) and momenta \(k_j\) introduced  in Eq.~\eqref{eq5.11},
\begin{equation}
	\begin{aligned}
		&\frac{\mathrm{d}\langle u^1\rangle}{\mathrm{d}u^0}
		=
		\frac{2}{\lambda}\langle\hat p^{\,1}\rangle,
		\qquad
		\frac{\mathrm{d}\langle u^2\rangle}{\mathrm{d}u^0}
		=
		\frac{2}{\lambda}\langle\hat p^{\,2}\rangle,\\
		&\frac{\mathrm{d}\langle\hat p^{\,1}\rangle}{\mathrm{d}u^0}
		=
		0,
		\qquad
		\frac{\mathrm{d}\langle\hat p^{\,2}\rangle}{\mathrm{d}u^0}
		=
		0.
	\end{aligned}
	\label{eq5.19}
\end{equation}
Consider a packet with the mean transverse momentum $\bm k=(k_1,k_2)$, where $\bm k$ is the same parameter as in the free plane wave in Eq.~\eqref{eq5.10}, i.e.,
$\langle\hat p^{\,j}\rangle=k_j$, $j=1,2$.
Integrating Eq.~\eqref{eq5.19} gives
\begin{equation}
	\langle u^j\rangle(u^0)
	=
	u_0^j+\frac{2k_j}{\lambda}u^0,
	\qquad
	u_0^j=\langle u^j\rangle(0),
	\qquad
	j=1,2.
	\label{eq5.20}
\end{equation}

Let us verify the free limit of the classical trajectory in Eq.~\eqref{eq5.7}. Equations~\eqref{eq5.7} and~\eqref{eq5.8} imply
\begin{equation}
	|\omega_c|R
	=
	\sqrt{(\dot u_0^1)^2+(\dot u_0^2)^2},
	\qquad
	\dot u_0^j=\dot u^j(0).
	\label{eq5.21}
\end{equation}
As $\omega_c\to0$, at fixed $u_0^j,\dot u_0^j$, the expansion of Eq.~\eqref{eq5.7} gives
\begin{equation}
	u^j(u^0)
	=
	u_0^j+\dot u_0^j u^0
	+O\!\left(|\omega_c|(u^0)^2\right),
	\qquad
	j=1,2,
	\label{eq5.22}
\end{equation}
uniformly in $u^0$ on each fixed bounded interval. Identifying the initial classical transverse momentum with $\bm k$, $p_0^j=k_j$, we have
$\dot u_0^j=2k_j/\lambda$.
Setting $u_0^j=\langle u^j\rangle(0)$, we find that the free limit of the classical trajectory in Eq.~\eqref{eq5.22} coincides with the motion of the wave packet center in Eq.~\eqref{eq5.20}.

Thus, for the particular implementation of the limit $B\to0$ considered
here, namely $G_N\to0^+$ at fixed $\lambda\neq0$ and $\Omega\neq0$, the reduced Hamiltonian tends to the free Hamiltonian. For a family of coherent states with fixed initial mean coordinates \(u_0^j\) and momenta \(k_j\), the parameters \(z_1,z_2\) are given in Eq.~\eqref{eq5.14}, and the renormalized functions \(\widetilde\psi_{\bm k,B}\) tend to free plane waves with transverse momentum \(\bm k\). For a localized free wave packet with the
same initial mean coordinates and momenta, the Ehrenfest Eqs.~\eqref{eq5.19} describe uniform rectilinear motion of its center. The limit $B\to0$ of the corresponding classical trajectory,
realized here through $G_N\to0^+$, coincides with this motion.

\section{Neutrino coherent states}\label{sec:coherent-classification}

In this section, we select a family of generalized coherent states with minimal spatial variance in the plane $(u^1,u^2)$ and adopt them as neutrino coherent states. We also obtain explicit expressions for the corresponding creation and annihilation operators acting in the space of bispinors in Eqs.~\eqref{eq6.16}--\eqref{eq6.20}.

\subsection{Coherent states}\label{subsec:explicit-coherent-states}

Let us introduce the following definition:

\begin{definition}\label{def:generalized-coherent-neutrino}
The collection of solutions $\Psi_{z_1,z_2,\lambda}$ defined by Eq.~\eqref{eq4.12} for
\begin{equation}
z_1,z_2,\rho_0\in\mathbb C,
\qquad |\rho_0|<1,
\qquad \lambda\in\mathbb R\setminus\{G_N\},
\label{eq6.1}
\end{equation}
is called a \emph{complete set of generalized neutrino coherent states}.
\end{definition}

For each fixed $\rho_0$, the subfamily parametrized by $z_1,z_2$ and $\lambda$ is overcomplete according to Lemma~\ref{lem:coherent-dirac-completeness} and provides the resolution of the identity in Eq.~\eqref{eq4.15}.

\begin{definition}\label{def:squeezed-coherent-neutrino}
Generalized neutrino coherent states with
\begin{equation}
\rho_0\neq0
\label{eq6.2}
\end{equation}
are called \emph{the squeezed neutrino coherent states}.
\end{definition}

From the family of generalized coherent states parametrized by the complex number $\rho_0$, $|\rho_0|<1$, let us select a complete family of states that minimize the Heisenberg uncertainty relation for the kinetic momenta:
$$
 \Delta\Pi_1\,\Delta\Pi_2\geq
\frac12\left|\langle [\hat\Pi_1,\hat\Pi_2]\rangle\right|=B/2.
$$
To this end, let us consider the coordinate variances $(\Delta u^j)^2=\langle (u^j-\langle u^j\rangle)^2\rangle$ in the plane $(u^1,u^2)$. From the inverse formulas in Eq.~\eqref{eq3.21}, the second eigenvalue equation in Eq.~\eqref{eq4.1}, and the inverse transformation in Eq.~\eqref{eq4.17}, we obtain
\begin{equation}
\begin{aligned}
&(\Delta u^1)^2=\frac{1+\operatorname{Re}\!\left(\rho_0\mathrm{e}^{-2\mathrm{i}\omega_a u^0}\right)}{B\left(1-|\rho_0|^2\right)},\qquad
(\Delta u^2)^2=\frac{1-\operatorname{Re}\!\left(\rho_0\mathrm{e}^{-2\mathrm{i}\omega_a u^0}\right)}{B\left(1-|\rho_0|^2\right)},\\
&\frac12\langle (u^1-\langle u^1\rangle)(u^2-\langle u^2\rangle)+(u^2-\langle u^2\rangle)(u^1-\langle u^1\rangle)\rangle
=-\frac{s\,\operatorname{Im}\!\left(\rho_0\mathrm{e}^{-2\mathrm{i}\omega_a u^0}\right)}{B\left(1-|\rho_0|^2\right)}.
\end{aligned}
\label{eq6.3}
\end{equation}
The expression for $(\Delta u^2)^2$  in Eq.~\eqref{eq6.3} coincides with Eq.~\eqref{eq4.20}. Let us introduce the covariance matrix of the coordinates $u^1,u^2$ with the components
\begin{equation}
	\Sigma^{ij}
	=
	\frac{1}{2}
	\left\langle
	\Delta u^i\Delta u^j+\Delta u^j\Delta u^i
	\right\rangle,
	\qquad
	\Delta u^i=u^i-\langle u^i\rangle,
	\qquad
	i,j=1,2.
	\label{eq6.4}
\end{equation}
According to Eq.~\eqref{eq6.3}, the eigenvalues of the matrix $\Sigma^{ij}$ are
\begin{equation}
	\frac{1}{B(1+|\rho_0|)},
	\qquad
	\frac{1}{B(1-|\rho_0|)}.
	\label{eq6.5}
\end{equation}
They are the coordinate variances along the principal axes of the spatial distribution. Hence, for $\rho_0\neq0$, the wave packet is squeezed along one principal direction and stretched along the other. The orientation of these directions is determined by the phase
$\rho_0\mathrm{e}^{-2\mathrm{i}\omega_a u^0}$.

Summing the first two equalities in Eq.~\eqref{eq6.3}, we obtain
\begin{equation}
(\Delta u^1)^2+(\Delta u^2)^2
=\frac{2}{B(1-|\rho_0|^2)}
\geq\frac{2}{B}.
\label{eq6.6}
\end{equation}
Equality is possible only for $|\rho_0|=0$, i.e., $\rho_0=0$. For the kinetic momenta, the inverse transformation in Eq.~\eqref{eq4.17} also implies
\begin{equation}
(\Delta\Pi_1)^2(\Delta\Pi_2)^2
=\frac{B^2}{4}\left[1+
\frac{4\left[\operatorname{Im}\!\left(\rho_0\mathrm{e}^{-2\mathrm{i}\omega_a u^0}\right)\right]^2}
{\left(1-|\rho_0|^2\right)^2}\right]
\geq\frac{B^2}{4}.
\label{eq6.7}
\end{equation}
Since $\omega_a\neq0$ for $B>0$ and $\lambda\neq G_N$, equality in the Heisenberg uncertainty relation holds at all $u^0$ if and only if $\rho_0=0$. For $\rho_0=0$, we have
\begin{equation}
(\Delta u^1)^2=(\Delta u^2)^2=\frac{1}{B},
\qquad
\frac12\langle 
(u^1-\langle u^1\rangle)(u^2-\langle u^2\rangle)
+(u^2-\langle u^2\rangle)(u^1-\langle u^1\rangle)\rangle=0,
\label{eq6.8}
\end{equation}
and for the kinetic momenta, where
$(\Delta\Pi_j)^2=\langle (\hat\Pi_j-\langle \hat\Pi_j\rangle)^2\rangle$,
\begin{equation}
(\Delta\Pi_1)^2=(\Delta\Pi_2)^2=\frac{B}{2},
\qquad
\Delta\Pi_1\,\Delta\Pi_2
=\frac12\left|\langle [\hat\Pi_1,\hat\Pi_2]\rangle\right|
=\frac{B}{2}.
\label{eq6.9}
\end{equation}
Moreover, the linear integral of motion~\eqref{eq3.32} takes the form
\begin{equation}
\hat A(u^0)=\mathrm{e}^{\mathrm{i}\omega_a u^0}\hat a.
\label{eq6.10}
\end{equation}
Thus, the covariance matrix $\Sigma^{ij}$ is time independent, whereas the wave packet center moves along the classical trajectory determined by the system in Eq.~\eqref{eq4.23} and Eq.~\eqref{eq5.7}.

\begin{definition}\label{def:coherent-neutrino}
The states in the subfamily of the complete set of generalized neutrino coherent states defined by the condition
\begin{equation}
\rho_0=0,
\label{eq6.11}
\end{equation}
are called \emph{neutrino coherent states}.
\end{definition}

For $\rho_0=0$, the function in Eq.~\eqref{eq4.11} simplifies. For $s=+1$, we have
\begin{subequations}\label{eq6.12}
\begin{equation}
\begin{aligned}
&\left.\psi_{z_1,z_2}(u^0,u^1,u^2)\right|_{\rho_0=0,\,s=+1}
=\sqrt{\frac{B}{2\pi}}\,\exp\Bigg\{
-\frac{B}{4}\left[(u^1)^2+(u^2)^2\right]
+\frac{\mathrm{i} B}{2}u^1u^2\\
&\qquad
-\frac{|z_1|^2+|z_2|^2}{2}
-\frac{\mathrm{i} G_N}{2}u^0
+z_1z_2\mathrm{e}^{-\mathrm{i}\omega_a u^0}\\
&\qquad
+\sqrt{\frac{B}{2}}\left[
 z_1\mathrm{e}^{-\mathrm{i}\omega_a u^0}(\mathrm{i} u^1-u^2)
+z_2(\mathrm{i} u^1+u^2)
\right]
\Bigg\}.
\end{aligned}
\label{eq6.12a}
\end{equation}
For $s=-1$, we have, respectively,
\begin{equation}
\begin{aligned}
&\left.\psi_{z_1,z_2}(u^0,u^1,u^2)\right|_{\rho_0=0,\,s=-1}
=\sqrt{\frac{B}{2\pi}}\,\exp\Bigg\{
-\frac{B}{4}\left[(u^1)^2+(u^2)^2\right]
-\frac{\mathrm{i} B}{2}u^1u^2\\
&\qquad
-\frac{|z_1|^2+|z_2|^2}{2}
-\mathrm{i}\left(\frac{G_N}{2}+\omega_a\right)u^0
+z_1z_2\mathrm{e}^{-\mathrm{i}\omega_a u^0}\\
&\qquad
+\sqrt{\frac{B}{2}}\left[
 z_1\mathrm{e}^{-\mathrm{i}\omega_a u^0}(\mathrm{i} u^1+u^2)
+z_2(\mathrm{i} u^1-u^2)
\right]
\Bigg\}.
\end{aligned}
\label{eq6.12b}
\end{equation}
Equations~\eqref{eq6.12a} and~\eqref{eq6.12b} are related by the reflection $u^2\to-u^2$ and an additional dynamical phase:
\begin{equation}
\left.\psi_{z_1,z_2}(u^0,u^1,u^2)\right|_{\rho_0=0,\,s=-1}
=\mathrm{e}^{-\mathrm{i}\omega_a u^0}
\left.\psi_{z_1,z_2}(u^0,u^1,-u^2)\right|_{\rho_0=0,\,s=+1}.
\label{eq6.12c}
\end{equation}
\end{subequations}
Recall that
$B=|G_N\Omega|/\sqrt2$, $s=\operatorname{sgn}(G_N\Omega)$,
$\omega_a=-\sqrt2|G_N\Omega|/(G_N-\lambda)$, while the quantity $\varepsilon_0$ in Eq.~\eqref{eq3.17} can be written as
$\varepsilon_0=[G_N+(1-s)\omega_a]/2$.

The corresponding set of solutions of the original Dirac Eq.~\eqref{eq2.10} is obtained from Eq.~\eqref{eq4.12} by evaluating the action of the differential operator in the first component on the functions in Eqs.~\eqref{eq6.12a} and~\eqref{eq6.12b}. For $s=+1$, we obtain
\begin{equation}
\begin{aligned}
&\left.\Psi_{z_1,z_2,\lambda}(u)\right|_{\rho_0=0,\,s=+1}
=\frac{\mathrm{e}^{-\mathrm{i}\lambda u^3/2}}{\sqrt{4\pi\sqrt2}}
\begin{pmatrix}
\dfrac{2\sqrt B}{\lambda-G_N}\,
 z_1\mathrm{e}^{-\mathrm{i}\omega_a u^0}\\[3mm]
1\\0\\0
\end{pmatrix}\\
&\qquad\times
\left.\psi_{z_1,z_2}(u^0,u^1,u^2)\right|_{\rho_0=0,\,s=+1},
\qquad \lambda\neq G_N.
\end{aligned}
\label{eq6.13}
\end{equation}
Equation~\eqref{eq6.13} contains the function in Eq.~\eqref{eq6.12a}. For $s=-1$, the action of the operator in the first component of the full bispinor on the function in Eq.~\eqref{eq6.12b} gives
\begin{equation}
\begin{aligned}
&\left.\Psi_{z_1,z_2,\lambda}(u)\right|_{\rho_0=0,\,s=-1}
=\frac{\mathrm{e}^{-\mathrm{i}\lambda u^3/2}}{\sqrt{4\pi\sqrt2}}
\begin{pmatrix}
\dfrac{2\sqrt B}{\lambda-G_N}\left[
 z_2+\sqrt{\dfrac{B}{2}}(\mathrm{i} u^1+u^2)
\right]\\[3mm]
1\\0\\0
\end{pmatrix}\\
&\qquad\times
\left.\psi_{z_1,z_2}(u^0,u^1,u^2)\right|_{\rho_0=0,\,s=-1},
\qquad \lambda\neq G_N.
\end{aligned}
\label{eq6.14}
\end{equation}

\subsection{Creation and annihilation operators in the bispinor space}\label{subsec:full-bispinor-oscillators}

At fixed $s=\pm1$ and $\lambda\neq G_N$, let us rewrite Eqs.~\eqref{eq6.13}--\eqref{eq6.14} as follows:
\begin{equation}
\Psi_s[\psi](u)
=
\frac{\mathrm{e}^{-\mathrm{i}\lambda u^3/2}}{\sqrt{4\pi\sqrt2}}
\begin{pmatrix}
 c\,\hat D_s\\[1mm]
 1\\
 0\\
 0
\end{pmatrix}
\psi(u^0,u^1,u^2),
\qquad
c=\frac{2\sqrt B}{\lambda-G_N},
\qquad
\hat D_+=\hat a_+,
\qquad
\hat D_-=\hat a_-^\dagger.
\label{eq6.15}
\end{equation}

Equation~\eqref{eq6.15} defines a map from solutions of the reduced Eq.~\eqref{eq3.4} to the space of bispinor solutions of the original Dirac Eq.~\eqref{eq2.10} describing active neutrino states.

The annihilation and creation operators for the bispinors in Eq.~\eqref{eq6.15} can be chosen as
\begin{equation}
\begin{aligned}
&\hat a_+^{(4)}=\hat a_+ I_4,
&&
\hat a_+^{(4)\sharp}
=\hat a_+^\dagger I_4
+\frac{c}{4}(1+\gamma^5)\gamma^0(\gamma^1+\mathrm{i}\gamma^2),\\[2mm]
&\hat a_-^{(4)}
=\hat a_- I_4
-\frac{c}{4}(1+\gamma^5)\gamma^0(\gamma^1+\mathrm{i}\gamma^2),
&&
\hat a_-^{(4)\sharp}=\hat a_-^\dagger I_4,
\end{aligned}
\label{eq6.16}
\end{equation}
where $\sharp$ denotes the adjoint with respect to the scalar product in Eq.~\eqref{eq2.29}. In the chiral representation in Eq.~\eqref{eq2.25}, one gets that
\begin{equation}
	\frac14(1+\gamma^5)\gamma^0(\gamma^1+\mathrm{i}\gamma^2)
	=
	\begin{pmatrix}
		0&1&0&0\\
		0&0&0&0\\
		0&0&0&0\\
		0&0&0&0
	\end{pmatrix}.
	\label{eq6.17}
\end{equation}
Note that these creation and annihilation operators are related to $\hat a$ and $\hat a^\dagger$ as follows:
\begin{equation}
\begin{aligned}
&\hat a_s^{(4)}\Psi_s[\psi]=\Psi_s[\hat a_s\psi],
\qquad
\hat a_s^{(4)\sharp}\Psi_s[\psi]=\Psi_s[\hat a_s^\dagger\psi],\\
&[\hat a_s^{(4)},\hat a_s^{(4)\sharp}]=I_4,
\qquad s=\pm1.
\end{aligned}
\label{eq6.18}
\end{equation}
Indeed, for $s=+1$, the additional matrix term in $\hat a_+^{(4)\sharp}$ adds $c\psi$ to the first component and replaces $c\hat a_+^\dagger\hat a_+\psi$ by $c\hat a_+\hat a_+^\dagger\psi$. For $s=-1$, the additional term in $\hat a_-^{(4)}$ subtracts $c\psi$ and replaces $c\hat a_-\hat a_-^\dagger\psi$ by $c\hat a_-^\dagger\hat a_-\psi$. Hence, Eq.~\eqref{eq6.18} follows from
$[\hat a_s,\hat a_s^\dagger]=1$.

The $u^0$-dependent creation and annihilation operators that are integrals of motion of the original Dirac Eq.~\eqref{eq2.10} on the space of bispinors in Eq.~\eqref{eq6.15} have the form,
\begin{equation}
\begin{aligned}
&\hat A_s^{(4)}(u^0)
=\mathrm{e}^{\mathrm{i}\omega_a u^0}\hat a_s^{(4)},
\qquad
\hat A_s^{(4)\sharp}(u^0)
=\mathrm{e}^{-\mathrm{i}\omega_a u^0}\hat a_s^{(4)\sharp},\\
&[\hat A_s^{(4)}(u^0),\hat A_s^{(4)\sharp}(u^0)]=I_4.
\end{aligned}
\label{eq6.19}
\end{equation}

Since the second pair of creation and annihilation operators $\hat b$ and $\hat b^\dagger$ commutes with $\hat D_s$, the second pair in the bispinor space is proportional to the identity matrix: $\hat b^{(4)} = \hat b I_4$, $\hat b^{(4)^\dagger} = \hat b^\dagger I_4$. The bispinors in Eqs.~\eqref{eq6.13}--\eqref{eq6.14} obey the equations,
\begin{equation}
\hat A_s^{(4)}(u^0)\Psi_{z_1,z_2,\lambda}
=z_1\Psi_{z_1,z_2,\lambda},
\qquad
(\hat b I_4)\Psi_{z_1,z_2,\lambda}
=z_2\Psi_{z_1,z_2,\lambda},
\qquad s=\pm1.
\label{eq6.20}
\end{equation}

\section{The degenerate case \texorpdfstring{$\lambda=G_N$}{lambda = GN} and coherent states}\label{sec:singular-lambda}

Note that, for $\lambda=G_N$, Eq.~\eqref{eq3.3} is inapplicable. Hence, this case must be studied separately.

In this section, we consider the exceptional value $\lambda=G_N$ in the active left-chiral subsystem of the massless Dirac equation. Accordingly, both right-chiral bispinor components are set to zero below. A complete classification of free right-chiral solutions at $\lambda=G_N$ is not given here.

We first consider $s=+1$ and construct coherent states, normalized in $L^2(\mathbb R^2)$, of the second pair of creation and annihilation operators. Then, we consider $s=-1$, for which the solutions at $\lambda=G_N$ have zero norm with respect to the scalar product in Eq.~\eqref{eq2.29} on the hypersurface $u^0=\mathrm{const}$.

\subsection{Integration of the Dirac equation for \texorpdfstring{$\lambda=G_N$}{lambda = GN}}\label{subsec:singular-direct}

Let us write the solution of the original Dirac Eq.~\eqref{eq2.10} in the form,
\begin{equation}
\Psi(u)=
\exp\left(-\frac{\mathrm{i}\lambda}{2}u^3\right)
\begin{pmatrix}
\phi_L\\
\psi_L\\
0\\
0
\end{pmatrix}.
\label{eq7.1}
\end{equation}
Substituting Eq.~\eqref{eq7.1} into Eqs.~\eqref{eq2.14}--\eqref{eq2.15}, we obtain the system of equations,
\begin{align}
&(\lambda-G_N)\phi_L
+\left[\sqrt2(\partial_{u^2}+\mathrm{i}\partial_{u^1})
+G_N\Omega u^2\right]\psi_L=0, \label{eq7.2}\\
&\left[-\sqrt2(\partial_{u^2}-\mathrm{i}\partial_{u^1})
+G_N\Omega u^2\right]\phi_L
+\left(2\mathrm{i}\partial_{u^0}-G_N\right)\psi_L=0. \label{eq7.3}
\end{align}
For $\lambda\neq G_N$, Eq.~\eqref{eq7.2} reproduces Eq.~\eqref{eq3.3}. For
\begin{equation}
\lambda=G_N
\label{eq7.4}
\end{equation}
Equation~\eqref{eq7.2} becomes an independent differential equation for $\psi_L$.

Let us first consider $s=+1$. Then $\beta=B>0$. The definitions in Eqs.~\eqref{eq3.15}--\eqref{eq3.16} directly imply the identities,
\begin{equation}
\begin{aligned}
&+\sqrt2(\partial_{u^2}+\mathrm{i}\partial_{u^1})+G_N\Omega u^2
=-2\sqrt B\,\hat a,\\
&-\sqrt2(\partial_{u^2}-\mathrm{i}\partial_{u^1})+G_N\Omega u^2
=-2\sqrt B\,\hat a^\dagger.
\end{aligned}
\label{eq7.5}
\end{equation}
Hence, under the condition in Eq.~\eqref{eq7.4}, the system in Eqs.~\eqref{eq7.2}--\eqref{eq7.3} takes the form
\begin{equation}
\hat a\psi_L=0,
\qquad
-2\sqrt B\,\hat a^\dagger\phi_L
+\left(2\mathrm{i}\partial_{u^0}-G_N\right)\psi_L=0.
\label{eq7.6}
\end{equation}

The following lemma holds.

\begin{lemma}\label{lem:singular-plus-case}
Let $s=+1$, $\lambda=G_N$ and, for each $u^0$, let the functions $\phi_L$ and $\psi_L$ belong to the domains of the creation and annihilation operators $\hat{a},\hat{a}^\dagger$ in $L^2(\mathbb R^2)$. Then all finite-norm solutions of the system in Eq.~\eqref{eq7.6} have the form,
\begin{equation}
\phi_L=0,
\qquad
\psi_L(u^0,u^1,u^2)
=\mathrm{e}^{-\mathrm{i}G_Nu^0/2}\,\chi(u^1,u^2),
\qquad
\hat a\chi=0.
\label{eq7.7}
\end{equation}
\end{lemma}

\begin{proof}
The first equation in Eq.~\eqref{eq7.6} and the independence of $\hat a$ from $u^0$ imply $\hat a\partial_{u^0}\psi_L=0$. Hence, $(2\mathrm{i}\partial_{u^0}-G_N)\psi_L\in\ker\hat a$. On the other hand, for any $\eta\in\ker\hat a$ and admissible $\phi_L$, we have
\begin{equation}
\left\langle \eta,\hat a^\dagger\phi_L\right\rangle
=\left\langle \hat a\eta,\phi_L\right\rangle=0.
\label{eq7.8}
\end{equation}
Hence, the two terms in the second equation in Eq.~\eqref{eq7.6} belong to mutually orthogonal subspaces and should vanish separately,
\begin{equation}
\left(2\mathrm{i}\partial_{u^0}-G_N\right)\psi_L=0,
\qquad
\hat a^\dagger\phi_L=0.
\label{eq7.9}
\end{equation}
The first equality gives the time dependence in Eq.~\eqref{eq7.7}. Since the operator \(\hat a^\dagger\) has a trivial kernel in \(L^2(\mathbb R^2)\), we have \(\phi_L=0\).
\end{proof}

Substituting Eq.~\eqref{eq7.7} into Eq.~\eqref{eq7.1} and taking into account $\lambda=G_N$, we obtain the bispinor
\begin{equation}
	\Psi(u)
	=
	\exp\left[
	-\frac{\mathrm{i}G_N}{2}(u^0+u^3)
	\right]
	\begin{pmatrix}
		0\\
		\chi(u^1,u^2)\\
		0\\
		0
	\end{pmatrix},
	\qquad
	\hat a\chi=0.
	\label{eq7.10}
\end{equation}
Thus, for $s=+1$ and $\lambda=G_N$, the first bispinor component vanishes, and its dependence on $(u^1,u^2)$ is fully determined by the function $\chi\in\ker\hat a$.
In the oscillator basis $\varphi_{n\ell}$, the condition $\hat a\chi=0$ means that the quantum number $n=0$. Hence, the general normalized solution can be represented as
\begin{equation}
	\chi(u^1,u^2)
	=
	\sum_{\ell=0}^{\infty}
	c_\ell\,\varphi_{0\ell}(u^1,u^2),
	\qquad
	\sum_{\ell=0}^{\infty}|c_\ell|^2=1.
	\label{eq7.11}
\end{equation}

Since the operators $\hat b,\hat b^\dagger$ commute with $\hat a$, they preserve the subspace $\ker\hat a$. Therefore, the condition $\hat a\chi=0$ does not fix the quantum number $\ell$. In the next subsection, we select from these solutions the states that are eigenstates of $\hat b$.

\subsection{Coherent states of the second pair of creation and annihilation operators}\label{subsec:singular-coherent}

Let us introduce the following definition.

\begin{definition}\label{def:singular-coherent}
For $s=+1$ and $\lambda=G_N$, a \emph{coherent state} is a solution in Eq.~\eqref{eq7.10} that is an eigenfunction of the operator $\hat b$:
\begin{equation}
\hat a\Psi_z(u)=0,
\qquad
\hat b\Psi_z(u)=z\Psi_z(u),
\qquad z\in\mathbb C.
\label{eq7.12}
\end{equation}
\end{definition}

According to Eq.~\eqref{eq7.10}, the function \(\chi = \chi_z\) satisfies Eq.~\eqref{eq7.13},
\begin{equation}
\hat a\chi_z=0,
\qquad
\hat b\chi_z=z\chi_z.
\label{eq7.13}
\end{equation}
Integrating the system in Eq.~\eqref{eq7.13}, we obtain
\begin{align}
	&\chi_z(u^1,u^2) =
	\exp\left(-\frac{|z|^2}{2}+\mathrm{i}z\bar w\right)\varphi_{00}(u^1,u^2),
	\quad
	\bar w=\sqrt{\frac{B}{2}}\left(u^1-\mathrm{i}u^2\right), \label{eq7.14}\\
	&\int_{\mathbb R^2} |\chi_z|^2\,\mathrm{d}u^1\mathrm{d}u^2=1,\quad
	 \left|\langle\chi_z,\chi_{z'}\rangle\right| = \exp\left(-\frac{|z-z'|^2}{2}\right), \label{eq7.15}\\
	&\int_{\mathbb C}\frac{d^2z}{\pi}|\chi_z\rangle\langle\chi_z| =
	\widehat P_{\ker\hat a}.\label{eq7.16}
\end{align}
where \(\varphi_{00}\) is the vacuum state of the operators \(\hat a\) and \(\hat b\), defined in Eq.~\eqref{eq3.23}, and $\widehat P_{\ker\hat a}$ is the projector onto the subspace $\ker\hat a$.

The corresponding solution of the original Dirac Eq.~\eqref{eq2.10} has the form,
\begin{equation}
	\Psi_{z}(u)
	=
	\frac{1}{\sqrt{4\pi\sqrt2}}
	\exp\left[
	-\frac{\mathrm{i}G_N}{2}(u^0+u^3)
	\right]
	\begin{pmatrix}
		0\\
		1\\
		0\\
		0
	\end{pmatrix}
	\chi_z(u^1,u^2).
	\label{eq7.17}
\end{equation}
The family in Eq.~\eqref{eq7.17} is not complete in the space of solutions of the original Dirac Eq.~\eqref{eq2.10}, since the spectral parameter is fixed by \(\lambda=G_N\). Each solution in this family contains a single Fourier mode in \(u^3\) and is not localized in the longitudinal direction.

Since one has that
\begin{equation}
	u^0+u^3=\sqrt2\,x^0,
	\qquad
	\exp\left[
	-\frac{\mathrm{i}G_N}{2}(u^0+u^3)
	\right]
	=
	\exp\left(
	-\frac{\mathrm{i}G_N}{\sqrt2}x^0
	\right),
	\label{eq7.18}
\end{equation}
the solutions in Eq.~\eqref{eq7.17} are independent of the longitudinal coordinate $x^3$. Moreover,
\begin{equation}
	\mathrm{i}\partial_{x^0}\Psi_{z}
	=
	\frac{G_N}{\sqrt2}\Psi_{z},
	\qquad
	p^3=0.
	\label{eq7.19}
\end{equation}
Thus, the states in Eq.~\eqref{eq7.17} are stationary states with energy $E=G_N/\sqrt2$.

For the solutions in Eq.~\eqref{eq7.17}, we have $\langle\hat a\rangle=0$ and $\langle\hat b\rangle=z$. Hence, the inverse relations in Eq.~\eqref{eq3.21} for $s=+1$ give
\begin{equation}
	\begin{aligned}
		\langle u^1\rangle
		&=
		-\sqrt{\frac{2}{B}}\,\mathrm{Im}\,z,
		\qquad
		\langle u^2\rangle
		=
		\sqrt{\frac{2}{B}}\,\mathrm{Re}\,z,
		\\
		(\Delta u^1)^2
		&=
		(\Delta u^2)^2
		=
		\frac{1}{B}.
	\end{aligned}
	\label{eq7.20}
\end{equation}
The mean coordinates and their variances in Eq.~\eqref{eq7.20} are independent of $u^0$. Thus, the center of the coherent state in the plane $(u^1,u^2)$ remains at rest, and its spatial width is constant.

For $s=-1$, $\beta=-B$, the operators in Eqs.~\eqref{eq7.2}--\eqref{eq7.3} are expressed in terms of the creation and annihilation operators in the reverse order:
\begin{equation}
	\begin{aligned}
		&+\sqrt2
		\left(
		\partial_{u^2}+\mathrm{i}\partial_{u^1}
		\right)
		+G_N\Omega u^2
		=
		-2\sqrt B\,\hat a^\dagger,
		\\
		&-\sqrt2
		\left(
		\partial_{u^2}-\mathrm{i}\partial_{u^1}
		\right)
		+G_N\Omega u^2
		=
		-2\sqrt B\,\hat a.
	\end{aligned}
	\label{eq7.21}
\end{equation}
Hence, for $\lambda=G_N$, the system in Eqs.~\eqref{eq7.2}--\eqref{eq7.3} takes the form,
\begin{subequations}
	\label{eq7.22}
	\begin{align}
		&\hat a^\dagger\psi_L=0,
		\label{eq7.22a}
		\\
		&-2\sqrt B\,\hat a\phi_L
		+
		\left(
		2\mathrm{i}\partial_{u^0}-G_N
		\right)\psi_L
		=
		0.
		\label{eq7.22b}
	\end{align}
\end{subequations}
The operator $\hat a^\dagger$ has no nontrivial kernel in $L^2(\mathbb R^2)$. Therefore,
\begin{equation}
	\psi_L=0,
	\qquad
	\hat a\phi_L=0.
	\label{eq7.23}
\end{equation}

Thus, the general solution of the active subsystem with finite norm in the plane $(u^1,u^2)$ has the form
\begin{equation}
	\begin{aligned}
		\Psi(u)
		&=
		\exp\left(
		-\frac{\mathrm{i}G_N}{2}u^3
		\right)
		\begin{pmatrix}
			1\\
			0\\
			0\\
			0
		\end{pmatrix} \phi_L(u^0,u^1,u^2),
		\\
		\phi_L(u^0,u^1,u^2)
		&=
		\sum_{\ell=0}^{\infty}
		c_\ell(u^0)\,
		\varphi_{0\ell}(u^1,u^2),
		\quad
		\sum_{\ell=0}^{\infty}|c_\ell(u^0)|^2<\infty,\quad
		\hat a\varphi_{0\ell}(u^1,u^2)=0.
	\end{aligned}
	\label{eq7.24}
\end{equation}
The coefficients $c_\ell(u^0)$ are arbitrary functions of $u^0$, since, after substituting $\psi_L=0$, the system in Eq.~\eqref{eq7.22} contains no equation determining the dependence of $\phi_L$ on $u^0$.

As in the case $s=+1$, one can formally select coherent solutions associated with the second pair of creation and annihilation operators in the subspace $\ker\hat a$. Let us set
\begin{equation}
	\chi_z^{(-)}(u^1,u^2)
	=
	\exp\left(
	-\frac{|z|^2}{2}
	+\mathrm{i}z\bar w
	\right)
	\varphi_{00}(u^1,u^2),
	\qquad
	\bar w
	=
	\sqrt{\frac{B}{2}}
	\left(
	u^1+\mathrm{i}u^2
	\right).
	\label{eq7.25}
\end{equation}
Then
\begin{equation}
	\hat a\chi_z^{(-)}=0,
	\qquad
	\hat b\chi_z^{(-)}=z\chi_z^{(-)},
	\qquad
	\int_{\mathbb R^2}
	|\chi_z^{(-)}|^2\,
	\mathrm{d}u^1\mathrm{d}u^2
	=
	1.
	\label{eq7.26}
\end{equation}
The corresponding subfamily of solutions in Eq.~\eqref{eq7.24} is obtained for
\begin{equation}
	\phi_{L,z}(u^0,u^1,u^2)
	=
	C_z(u^0)\,
	\chi_z^{(-)}(u^1,u^2),
	\label{eq7.27}
\end{equation}
where $C_z(u^0)$ is an arbitrary function.

For the solutions in Eq.~\eqref{eq7.24}, we have $\Psi_{(-)}=0$. Hence, the definition of the scalar product in Eq.~\eqref{eq2.29} directly implies
\begin{equation}
	(\Psi,\Psi)_{u^0}=0.
	\label{eq7.28}
\end{equation}
Thus, all nonzero solutions in Eq.~\eqref{eq7.24}, including the coherent subfamily in Eq.~\eqref{eq7.27}, have zero norm with respect to the scalar product in Eq.~\eqref{eq2.29}.

Hence, the solutions in Eq.~\eqref{eq7.24} belong to the null space of the positive-semidefinite form in Eq.~\eqref{eq2.29}. Upon taking the quotient of the solution space by this null space, they represent the zero vector and do not define physical active neutrino states on the hypersurface $u^0=\mathrm{const}$.

Thus, for $\lambda=G_N$, the cases $s=+1$ and $s=-1$ differ substantially. For $s=+1$, the component $\psi_L\in\ker\hat a$ is nonzero, and normalized coherent states associated with the second pair of creation and annihilation operators exist. For $s=-1$, the condition $\ker\hat a^\dagger=\{0\}$ implies $\psi_L=0$, whereas the nonzero component $\phi_L\in\ker\hat a$ does not contribute to the norm in Eq.~\eqref{eq2.29}. The corresponding coherent subfamily consists of formal coherent solutions in the null space rather than physical states.

\section{The parameters \texorpdfstring{$\lambda$}{lambda} and \texorpdfstring{$m_{\mathrm{eff}}$}{meff} for neutrinos in NS}\label{sec:lambda-neutron-star}

In this section, we determine the possible values of the spectral parameter $\lambda$ and estimate the parameter $m_{\mathrm{eff}}$ for the coherent solutions constructed above, given the physical parameters of the neutrino and neutron star matter.

We consider ordinary coherent states, $\rho_0=0$. We also choose the coordinate orientation so that $\Omega>0$. For $G_N>0$, this corresponds to $s=+1$. The case $s=-1$ is not considered in this section.

Since a coherent state with $z_1\neq0$ is not stationary, we use the mean physical energy below. We first determine the relation between the mean energy and the spectral parameter $\lambda$. We then estimate $G_N$, $B$, and $m_{\mathrm{eff}}$ for matter with a density of the order of nuclear matter density.

For numerical estimates, we consider a small region where space can be regarded as flat, and the matter parameters and velocity field vary slowly on the scale of the coherent wave packet.

\subsection{Mean physical energy and possible values of \texorpdfstring{$\lambda$}{lambda}}\label{subsec:mean-energy-lambda}

For a solution with fixed $\lambda$, we have $2\widehat{\widetilde p}_3\Psi=\lambda\Psi$. The transformations in Eq.~\eqref{eq2.4} imply
\begin{equation}
\hat p_0
=\frac{1}{\sqrt2}\left(
\mathrm{i}\partial_{u^0}+\mathrm{i}\partial_{u^3}
\right),
\qquad
\hat p^{\,3}
=\frac{1}{\sqrt2}\left(
\mathrm{i}\partial_{u^0}-\mathrm{i}\partial_{u^3}
\right).
\label{eq8.1}
\end{equation}
Thus, for a coherent solution normalized in the plane $(u^1,u^2)$ and delta-normalized with respect to $\lambda$, the mean values below are defined at fixed $\lambda$ using the normalized transverse function: the common longitudinal delta-normalization factor in the numerator and denominator cancels. Let us define
\begin{equation}
\overline E\equiv\langle\hat p_0\rangle
=\frac{1}{\sqrt2}\left(\langle\hat H\rangle+\frac{\lambda}{2}\right),
\qquad
\overline p^{\,3}\equiv\langle\hat p^{\,3}\rangle
=\frac{1}{\sqrt2}\left(\langle\hat H\rangle-\frac{\lambda}{2}\right).
\label{eq8.2}
\end{equation}
From Eq.~\eqref{eq8.2}, we obtain the exact relation,
\begin{equation}
\lambda
=\sqrt2\left(\overline E-\overline p^{\,3}\right).
\label{eq8.3}
\end{equation}

To express $\overline E$ in terms of the coherent-state parameters, we use the inverse transformation in Eq.~\eqref{eq4.17}. The first equation  in Eq.~\eqref{eq4.1} directly implies
\begin{equation}
\overline n
\equiv\left\langle\hat a^\dagger\hat a\right\rangle
=|z_1|^2.
\label{eq8.4}
\end{equation}
For $s=+1$, Eq.~\eqref{eq3.17} gives
\begin{equation}
\langle\hat H\rangle
=\frac{G_N}{2}+\omega_a\overline n,
\qquad
\omega_a=\frac{2B}{\lambda-G_N}.
\label{eq8.5}
\end{equation}
Substituting Eq.~\eqref{eq8.5} into Eq.~\eqref{eq8.2}, we obtain the relation between the mean physical energy and the spectral parameter $\lambda$ for $\lambda\neq G_N$:
\begin{equation}
Q = \sqrt2\,\overline E-G_N = 
\frac{\lambda-G_N}{2}
+\frac{2B\overline n}{\lambda-G_N}.
\label{eq8.6}
\end{equation}
If $\overline n>0$, multiplying Eq.~\eqref{eq8.6} by $\lambda-G_N$, leads to the quadratic Eq.~\eqref{eq8.7}
\begin{equation}
(\lambda-G_N)^2
-2Q(\lambda-G_N)
+4B\overline n=0.
\label{eq8.7}
\end{equation}
Hence, under the condition,
\begin{equation}
	Q^2\geq4B\overline n,
	\label{eq8.8}
\end{equation}
Eq.~\eqref{eq8.7} has real roots
\begin{equation}
\lambda_\pm
=
G_N+Q
\pm
\sqrt{Q^2-4B\overline n},
\label{eq8.9}
\end{equation}
For $Q^2>4B\overline n$, the roots in Eq.~\eqref{eq8.9} are distinct, whereas for $Q^2=4B\overline n$ they coincide, giving a single double root.

For $Q\neq0$, the physical meaning of the two roots is determined by their behavior as $B\to0$. Let us denote
\begin{equation}
\begin{aligned}
\lambda_{\mathrm{reg}}
&=G_N+Q+\operatorname{sgn}(Q)
\sqrt{Q^2-4B\overline n}\rightarrow G_N+2Q,\quad B\to 0, \\
\lambda_{\mathrm{near}}
&=G_N+Q-\operatorname{sgn}(Q)
\sqrt{Q^2-4B\overline n}\longrightarrow G_N,\quad B\to 0.
\end{aligned}
\label{eq8.10}
\end{equation}
As the vorticity vanishes, the root $\lambda_{\mathrm{reg}}$ continuously approaches the regular solution with $\lambda\neq G_N$. As $B\to0$, the root $\lambda_{\mathrm{near}}$ approaches the exceptional value $\lambda=G_N$ in the spectral parameter. This statement concerns only the spectral parameter: the convergence of the corresponding coherent solutions to the degenerate solutions of Sec.~\ref{sec:singular-lambda} requires a separate study of the singular limit $\lambda-G_N=O(B)$.

According to Eq.~\eqref{eq8.3}, the two roots correspond to different mean longitudinal momenta
$\overline p^{\,3}_{\mathrm{reg,near}}
=\overline E-\lambda_{\mathrm{reg,near}}/\sqrt2$. Thus, at the same mean energy, the roots describe different longitudinal kinematic regimes.

Let us separately consider $\overline n=0$. If $Q\neq0$, the original equality in Eq.~\eqref{eq8.6} for $\lambda\neq G_N$ gives the unique regular value
\begin{equation}
\lambda=G_N+2Q=2\sqrt2\,\overline E-G_N,
\qquad \overline n=0,\quad Q\neq0.
\label{eq8.11}
\end{equation}
If $\overline n=0$ and $Q=0$, no regular solution with $\lambda\neq G_N$ exists: Eq.~\eqref{eq8.11} gives $\lambda=G_N$. The point $Q=0$, i.e., $\overline E=G_N/\sqrt2$, for $\overline n=0$ belongs to the degenerate case $\lambda=G_N$ studied in Sec.~\ref{sec:singular-lambda}.

For $Q\neq0$ and $Q^2\gg4B\overline n$, the roots in Eq.~\eqref{eq8.10} have the form,
\begin{equation}
\begin{aligned}
\lambda_{\mathrm{reg}}
&=G_N+2Q
-\frac{2B\overline n}{Q}
+O\!\left(\frac{B^2\overline n^2}{|Q|^3}\right),\\
\lambda_{\mathrm{near}}
&=G_N
+\frac{2B\overline n}{Q}
+O\!\left(\frac{B^2\overline n^2}{|Q|^3}\right).
\end{aligned}
\label{eq8.12}
\end{equation}
Thus, at the same mean energy and under the strict inequality $Q^2>4B\overline n$, the regular equation for $\lambda\neq G_N$ gives two distinct values of the spectral parameter. For the astrophysical estimate, we use $\lambda=\lambda_{\mathrm{reg}}$, which remains separated from the exceptional value $G_N$ and tends, as $B\to0$, to the regular solution of the problem with zero vorticity.

The condition in Eq.~\eqref{eq8.8} for real roots can be expressed directly in terms of the mean physical energy. Since $Q=\sqrt2\,\overline E-G_N$, for $\overline n>0$ it is equivalent to
\begin{equation}
	\left|
	\overline E-\frac{G_N}{\sqrt2}
	\right|
	\geq
	\sqrt{2B\overline n}.
	\label{eq8.13}
\end{equation}
Hence, for given $B$ and $\overline n>0$, two distinct real values of the spectral parameter, $\lambda_{\mathrm{reg}}$ and $\lambda_{\mathrm{near}}$, exist for $|\overline E-G_N/\sqrt2|>\sqrt{2B\overline n}$. At equality, the roots coincide. Consider the interval
\begin{equation}
	\left|
	\overline E-\frac{G_N}{\sqrt2}
	\right|
	<
	\sqrt{2B\overline n}.
	\label{eq8.14}
\end{equation}
Within this interval, the discriminant of Eq.~\eqref{eq8.7} is negative; hence, the regular Eq.~\eqref{eq8.6} for $\lambda\neq G_N$ has no real roots in $\lambda$. For $\overline n>0$, the value $\overline E=G_N/\sqrt2$ does not belong to the regular family. The same numerical value of the energy is implemented by a separate family of normalized degenerate solutions for $s=+1$ and $\lambda=G_N$, considered in Sec.~\ref{sec:singular-lambda}; for these solutions, $\hat a\chi=0$ and, hence, $\overline n=0$.

\subsection{Estimate of \texorpdfstring{$m_{\mathrm{eff}}$}{meff} for neutrinos in NS}\label{subsec:lambda-neutron-star-estimate}

To get the estimate of the effective neutrino mass, let us choose the following parameters of dense neutron matter: neutron density $n_n\sim10^{38}\,\mathrm{cm}^{-3}$, the corresponding weak potential $G_N/\sqrt2\simeq7\,\mathrm{eV}$, a local velocity gradient $\Omega\simeq10^3\,\mathrm{s}^{-1}$, and stellar radius $R_*=10\,\mathrm{km}$. Then
\begin{equation}
\begin{aligned}
&n_n\sim10^{38}\,\mathrm{cm}^{-3},
\qquad
G_N\simeq9.90\,\mathrm{eV},\\
&\Omega\simeq6.58\times10^{-13}\,\mathrm{eV},
\qquad
R_*\simeq10\,\mathrm{km}.
\end{aligned}
\label{eq8.15}
\end{equation}
For these parameters,
\begin{equation}
B=\frac{G_N\Omega}{\sqrt2}
\simeq4.61\times10^{-12}\,\mathrm{eV}^2.
\label{eq8.16}
\end{equation}

Consider a coherent state with $\overline n=|z_1|^2=1$. For $\overline n=1$, the right-hand side of Eq.~\eqref{eq8.13} with the chosen parameters is approximately $3.0\times10^{-6}\,\mathrm{eV}$. Thus, two distinct real values of the spectral parameter, $\lambda_{\mathrm{reg}}$ and $\lambda_{\mathrm{near}}$, exist for $|\overline E-G_N/\sqrt2|>3.0\times10^{-6}\,\mathrm{eV}$. At equality, these parameters coincide.

For the neutron star parameters adopted above and $\overline n=1$, the
condition for real solutions of Eq.~\eqref{eq8.7} becomes
\begin{equation}
	\left|
	\overline E-\frac{G_N}{\sqrt2}
	\right|
	\geq
	3.0\times10^{-6}\,\mathrm{eV}.
\end{equation}
Hence, within the narrow interval
\begin{equation}
	\left|
	\overline E-\frac{G_N}{\sqrt2}
	\right|
	<
	3.0\times10^{-6}\,\mathrm{eV},
\end{equation}
the discriminant of Eq.~\eqref{eq8.7} is negative and the regular family
with $\lambda\neq G_N$ has no real value of the spectral parameter
$\lambda$.

The central value $\overline E=G_N/\sqrt2$ coincides with the energy of the normalized degenerate solutions for $s=+1$ and $\lambda=G_N$. However, these solutions have the quantum number $\overline n=0$ and belong to a separate family rather than the family with $\overline n=1$ considered here.

The characteristic oscillator length and the radius of motion of the wave packet center about the center specified by $z_2$ are
\begin{equation}
\ell_B=\frac{1}{\sqrt B}\simeq9.2\,\mathrm{cm},
\qquad
R_\perp=\sqrt{\frac{2\overline n}{B}}\simeq13\,\mathrm{cm}.
\label{eq8.19}
\end{equation}
Thus, this packet is localized on a scale much smaller than the NS radius. This condition controls the packet size rather than its position. We additionally assume that the parameter $z_2$ places the packet center within the local region under consideration and that, for a chosen $\kappa>0$, the condition in Eq.~\eqref{eq4.21} holds,
$|\Omega|\bigl(|\langle u^2\rangle|+\kappa\Delta u^2\bigr)<1$.
The length scales over which the matter density and velocity gradient vary must also be much greater than $\ell_B$ and $R_\perp$.

The effective mass of the classical model, $m_{\mathrm{eff}}=(\lambda-G_N)/2$, satisfies the equation 
\begin{equation}
	m_{\mathrm{eff}}^2-Qm_{\mathrm{eff}}+B\overline n=0.
	\label{eq8.20}
\end{equation} 
obtained from Eq.~\eqref{eq8.7}. Using the classification in Eq.~\eqref{eq8.10}, let us define
\begin{equation}
\begin{aligned}
m_{\mathrm{eff}}^{\mathrm{reg}}
&=\frac12\left[Q+\operatorname{sgn}(Q)
\sqrt{Q^2-4B\overline n}\right],\\
m_{\mathrm{eff}}^{\mathrm{near}}
&=\frac12\left[Q-\operatorname{sgn}(Q)
\sqrt{Q^2-4B\overline n}\right].
\end{aligned}
\label{eq8.21}
\end{equation}
We use the first root for the astrophysical estimate. In this case, it remains finite and does not approach the point $\lambda=G_N$ at $B\to0$.

As a specific low-energy state, let us take $\overline E=5\,\mathrm{eV}$. Then
\begin{equation}
Q=\sqrt2\,\overline E-G_N
\simeq-2.828\,\mathrm{eV},
\qquad
\frac{B\overline n}{Q^2}\simeq5.76\times10^{-13}.
\label{eq8.22}
\end{equation}
For ${B\overline n}/{Q^2}\ll1$, Eq.~\eqref{eq8.21} gives the following expansion for the regular root:
\[
m_{\mathrm{eff}}^{\mathrm{reg}}
=
Q-\frac{B\overline n}{Q}
+O\!\left(\frac{B^2\overline n^2}{|Q|^3}\right).
\]
Hence, the relative change of the regular root compared with its value at \(B=0\), $m_{\mathrm{eff}}^{\mathrm{reg}}\big|_{B=0}=Q$, is of order
\[
\frac{\left|m_{\mathrm{eff}}^{\mathrm{reg}}-Q\right|}{|Q|}
\simeq
\frac{B\overline n}{Q^2}
\simeq5.76\times10^{-13}.
\]
Numerically,
\begin{equation}
\begin{aligned}
m_{\mathrm{eff}}^{\mathrm{reg}}
&\simeq-2.828\,\mathrm{eV},
&\lambda_{\mathrm{reg}}&\simeq4.243\,\mathrm{eV},\\
m_{\mathrm{eff}}^{\mathrm{near}}
&\simeq-1.63\times10^{-12}\,\mathrm{eV},
&\lambda_{\mathrm{near}}&\simeq G_N.
\end{aligned}
\label{eq8.23}
\end{equation}

Typical values of the regular root $m_{\mathrm{eff}}^{\mathrm{reg}}$ at mean energies of the order of a few electronvolts are
\begin{equation}
\begin{aligned}
\overline E=5\,\mathrm{eV}:&\quad
m_{\mathrm{eff}}^{\mathrm{reg}}\simeq-2.83\,\mathrm{eV},\\
\overline E=6\,\mathrm{eV}:&\quad
m_{\mathrm{eff}}^{\mathrm{reg}}\simeq-1.41\,\mathrm{eV},\\
\overline E=8\,\mathrm{eV}:&\quad
m_{\mathrm{eff}}^{\mathrm{reg}}\simeq+1.41\,\mathrm{eV},\\
\overline E=10\,\mathrm{eV}:&\quad
m_{\mathrm{eff}}^{\mathrm{reg}}\simeq+4.24\,\mathrm{eV}.
\end{aligned}
\label{eq8.24}
\end{equation}
Thus, for the regular root at energies of the order of a few electronvolts, \(|m_{\mathrm{eff}}|\) is also of the order of a few electronvolts, except in the energy region near \(G_N/\sqrt2\), where \(Q\) becomes small.

Let us check the sensitivity of the regular root to a substantial increase in $\overline n$. To this end, we formally set the radius of motion of the wave packet center equal to the stellar scale, $R_\perp=10\,\mathrm{km}$. Using Eq.~\eqref{eq8.19} with the value of $B$ in Eq.~\eqref{eq8.16}, we obtain
\begin{equation}
\overline n\simeq5.9\times10^9,
\qquad
m_{\mathrm{eff}}^{\mathrm{reg}}(\overline E=5\,\mathrm{eV})
\simeq-2.82\,\mathrm{eV},
\qquad
m_{\mathrm{eff}}^{\mathrm{near}}
\simeq-9.7\times10^{-3}\,\mathrm{eV}.
\label{eq8.25}
\end{equation}
Even with this formal substitution, the regular root $m_{\mathrm{eff}}^{\mathrm{reg}}$ remains of the order of a few electronvolts.

Thus, for the chosen NS parameters, the regular root $m_{\mathrm{eff}}^{\mathrm{reg}}$ for low-energy coherent states is of the order of a few electronvolts. Small values $|m_{\mathrm{eff}}|\ll1\,\mathrm{eV}$ correspond to the other root, $m_{\mathrm{eff}}^{\mathrm{near}}$, which tends to zero as $B\to0$, so that $\lambda\to G_N$.

\begin{figure}[t]
	\centering
	\includegraphics[width=0.82\textwidth]{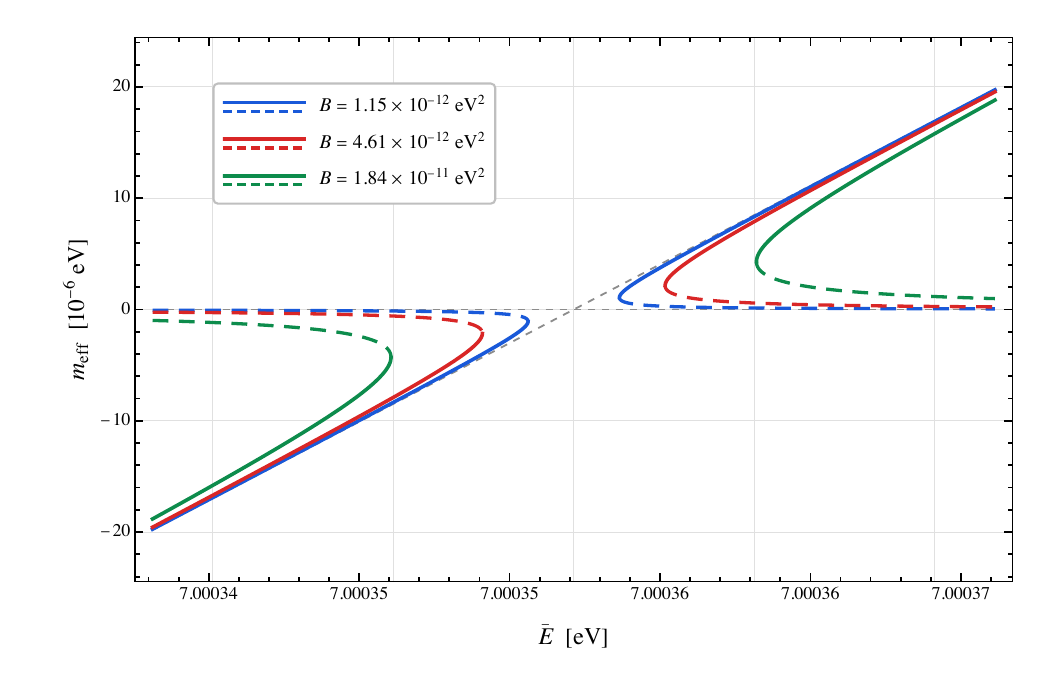}
	\caption{
		Dependence of the effective neutrino mass $m_{\mathrm{eff}}$ on the mean
		neutrino energy $\overline E$ for $G_N=9.90\,\mathrm{eV}$ and
		$\overline n=1$. Solid curves show
		$m_{\mathrm{eff}}^{\mathrm{reg}}$, while dashed curves represent
		$m_{\mathrm{eff}}^{\mathrm{near}}$. Different colors indicate different
		values of $B$. The gray dashed lines show the asymptotic values
		$m_{\mathrm{eff}}=\sqrt{2}\,\overline E-G_N$ and
		$m_{\mathrm{eff}}=0$.
	}
	\label{fig:meff-energy}
\end{figure}

To illustrate the dependence of the parameter $m_{\mathrm{eff}}=(\lambda-G_N)/2$ on the average energy $\overline E$, in Fig.~\ref{fig:meff-energy}, we show both $m_{\mathrm{eff}}^{\mathrm{reg}}$ and $m_{\mathrm{eff}}^{\mathrm{near}}$ for $G_N=9.90\,\mathrm{eV}$, at $\bar{n}=1$ and different 
values
of
the parameter $B = G_N \Omega / \sqrt{2}$.

The solid lines correspond to the regular root $m_{\mathrm{eff}}^{\mathrm{reg}}$, whereas the dashed ones represent the root $m_{\mathrm{eff}}^{\mathrm{near}}$, for which $\lambda\to G_N$ as $B\to 0$. Both branches have central symmetry with respect to the point $\left(G_N/\sqrt2,0\right)$.
For all considered values
of
$B$,
the regular branches approach the same asymptotic behavior  
$m_{\mathrm{eff}} \to \sqrt{2}\,\bar{E}-G_N$, while the parameter $m_{\mathrm{eff}}^{\mathrm{near}}$ approaches to zero: $m_{\mathrm{eff}} \to 0$.

While $B$ increases, the interval
\begin{equation*}
	\left|
	\overline E-\frac{G_N}{\sqrt2}
	\right|
	<
	\sqrt{2B\overline n},
\end{equation*}
where Eq.~\eqref{eq8.20} for $m_{\mathrm{eff}}$ has no real roots, becomes broader. 

Meanwhile, the parameter $m_{\mathrm{eff}}^{\mathrm{reg}}(\overline E)$ is practically
independent of $B$ far from this narrow region, whereas $m_{\mathrm{eff}}^{\mathrm{near}}$
rapidly approaches to zero. Thus, the impact of the matter rotation on the relation between
$\overline E$ and $m_{\mathrm{eff}}$ is the most sizable in a narrow
neighborhood of the value $\overline E=G_N/\sqrt2$, corresponding to
the special region $\lambda\approx G_N$.

\section{Conclusion}\label{sec:CONCL}

In this work we constructed the coherent states of neutrinos electroweakly interacting with a rotating background matter. Neutrinos are taken to be Dirac particles. The interaction with mater is considered in the forward scattering approximation. In this situation, we incorporate the matter effects into the external axial-vector field, which depends on spatial coordinates.

While constructing coherent states, we work in frames of
the relativistic single particle quantum mechanics. Therefore, a neutrino
is described by a $c$-number bispinor wave function obeying
the Dirac equation. This formulation of the problem differs from the pioneering work by Glauber in Ref.~\cite{Gla63}, where coherent states of a quantized electromagnetic field are defined in the bosonic Fock space as eigenstates of the annihilation operators.

For the matter velocity profile considered in our work,
the reduction the Dirac equation to the light-cone variables leads to
a reduced Hamiltonian that can be written in terms of the creation and annihilation operators of a harmonic oscillator. It allows us to construct
coherent neutrino states as the special solutions of the single-particle
Dirac equation. A similar approach was used in Refs.~\cite{BagrovBuchbinderGitman1976,GitmanBagrov2005, BagrovGavrilovGitmanMeira2011} while studying coherent states of relativistic particles in frames of the single-particle relativistic quantum mechanics.

In Sec.~\ref{sec:lf-reduction}, we have considered the Dirac equation for a neutrino interacting with matter having a coordinate dependent effective potential. We have rewritten this equation in the light cone variables. Using the projected equations~\eqref{eq2.14}--\eqref{eq2.15}, we have reduced the Dirac equation to a differential Eq.~\eqref{eq2.16} for a single projection, with the rest of the components being expressed through this projection. 

Then, we have approximated the rotating matter by the velocity field in Eq.~\eqref{eq2.30}. We have derived the two component Schr\"odinger Eq.~\eqref{eq2.32} which is valid for massive neutrinos in rotating matter. This equation couples active and sterile neutrinos through the mass term.

In Sec.~\ref{sec:oscillator-representation}, we consider massless particles. In this situation, the wave equations for active and sterile neutrinos decouple, as shown in Eq.~\eqref{eq3.8}. In the massless case, the effective Hamiltonian was shown in Eq.~\eqref{eq3.17} to be equivalent to a harmonic oscillator. We have derived its eigenfunctions and eigenvalues, which are used to built the coherent states.

In Sec.~\ref{sec:coherent-states}, we have constructed the generalized coherent states in case when $\lambda\neq G_N$. We have defined them as the eigenfunctions of the operators $\hat A(u^0)$ and $\hat b$. 
Then, In Eqs.~\eqref{eq4.18}-\eqref{eq4.22}, we have computed the mean values of the momentum and coordinates, as well as we have derived the Ehrenfest Eqs.~\eqref{eq4.23}, which determine the neutrino wave packet evolution.

In Sec.~\ref{sec:classical-limits}, we have provided the classical physics interpretation of the coherent states obtained. In particular, we have defined the effective neutrino mass owing to the electroweak interaction of massless neutrinos with a rotating matter. Then, we have also considered the situation of a slow rotation. In this case, the neutrino evolution is described by a plane wave solution. Moreover, the center of a neutrino wave packet moves along a classical trajectory.

In Sec.~\ref{sec:coherent-classification}, we have derived the unsqueezed  coherent states, which are defined as the neutrino coherent states. This kind a coherent states are characterized by reaching the lower bound in the Heisenberg uncertainty relation. We have also represented this coherent states in terms of bispinors. We have constructed the creation and annihilation operators for these coherent states in Eqs.~\eqref{eq6.16}--\eqref{eq6.20}.

We have considered the special case $\lambda=G_N$ in Sec.~\ref{sec:singular-lambda}. We have obtained the values of the observables in this case, as well as we have indicated which neutrino states can be implemented physically.

Finally, in Sec.~\ref{sec:lambda-neutron-star}, we have applied our results for the description of astrophysical neutrinos emitted by a rotating NS. In particular, we have evaluated the effective neutrino mass, introduced in Sec.~\ref{sec:classical-limits}, for different neutrino energies. It turns out to be in a few eV range for realistic values of NS parameters.

\appendix
\renewcommand{\theequation}{\thesection.\arabic{equation}}

\section{Proof of Lemma~\ref{lem:projected-reduction}}\label{app-A}

In this appendix, we prove Lemma~\ref{lem:projected-reduction}. We first write the Dirac Eq.~\eqref{eq2.10} in the form,
\begin{equation}
	\left(
	\widetilde\gamma^0\widehat{\widetilde p}_0
	+\widetilde\gamma^3\widehat{\widetilde p}_3
	+\widetilde\gamma^j\widehat{\widetilde p}_j
	+U-m_\nu
	\right)
	\left(\Psi_{(+)}+\Psi_{(-)}\right)=0,
	\qquad j=1,2.
	\label{eqA.1}
\end{equation}
We use the following identities:
\begin{equation}
	\begin{aligned}
		&\widehat P_{(+)}\widetilde\gamma^0=\widetilde\gamma^0,
		\qquad
		\widehat P_{(+)}\widetilde\gamma^3=0,
		\qquad
		\widehat P_{(-)}\widetilde\gamma^0=0,
		\qquad
		\widehat P_{(-)}\widetilde\gamma^3=\widetilde\gamma^3,\\
		&[\widehat P_{(\pm)},\widetilde\gamma^j]=0,
		\qquad
		\widetilde\gamma^0\Psi_{(+)}=0,
		\qquad
		\widetilde\gamma^3\Psi_{(-)}=0.
	\end{aligned}
	\label{eqA.2}
\end{equation}
Applying $\widehat P_{(+)}$ to Eq.~\eqref{eqA.1} from the left, we obtain that
\begin{equation}
	\widetilde\gamma^0\widehat{\widetilde p}_0\Psi_{(-)}
	+\left(\widetilde\gamma^j\widehat{\widetilde p}_j-m_\nu\right)\Psi_{(+)}
	+\widehat P_{(+)}U\Psi_{(-)}
	+\widehat P_{(+)}U\Psi_{(+)}=0.
\end{equation}
Let us multiply this equality from the left by $\widetilde\gamma^3/2$. Using
\begin{equation}
	\frac12\widetilde\gamma^3\widetilde\gamma^0=\widehat P_{(-)},
	\qquad
	\widehat P_{(-)}\Psi_{(-)}=\Psi_{(-)},
	\qquad
	\frac12\widetilde\gamma^3\widetilde\gamma^0
	\widehat{\widetilde p}_0\Psi_{(-)}=\widehat{\widetilde p}_0\Psi_{(-)},
\end{equation}
we obtain Eq.~\eqref{eq2.14}.

We now apply the projector $\widehat P_{(-)}$ to the original Eq.~\eqref{eqA.1}:
\begin{equation}
	\widetilde\gamma^3\widehat{\widetilde p}_3\Psi_{(+)}
	+\left(\widetilde\gamma^j\widehat{\widetilde p}_j-m_\nu\right)\Psi_{(-)}
	+\widehat P_{(-)}U\Psi_{(+)}
	+\widehat P_{(-)}U\Psi_{(-)}=0.
\end{equation}
Multiplication from the left by $\widetilde\gamma^0$ gives
\begin{equation}
	\widetilde\gamma^0\widetilde\gamma^3
	\widehat{\widetilde p}_3\Psi_{(+)}
	=2\widehat P_{(+)}\widehat{\widetilde p}_3\Psi_{(+)}
	=2\widehat{\widetilde p}_3\Psi_{(+)}.
\end{equation}
This implies the second projected Eq.~\eqref{eq2.15} and Eq.~\eqref{eq2.17} for $\Psi_{(+)}$. Substituting Eq.~\eqref{eq2.17} into Eq.~\eqref{eq2.14} gives the closed Eq.~\eqref{eq2.16} for $\Psi_{(-)}$. Conversely, let $\Psi_{(-)}$ satisfy Eq.~\eqref{eq2.16}, and let $\Psi_{(+)}$ be defined by Eq.~\eqref{eq2.17}. Then the second projected equation holds by construction, and the first coincides with Eq.~\eqref{eq2.16}. This proves Lemma~\ref{lem:projected-reduction}.

\end{document}